\documentclass[letterpaper,11pt,onecolumn,draftcls]{IEEEtran}  

\IEEEoverridecommandlockouts

\usepackage{epsfig} 
\usepackage{amsmath} 
\usepackage{amssymb}  

\usepackage{cite}      
\usepackage{subfigure} 
\usepackage{url}       

\usepackage{verbatim}   
\usepackage{tabularx}

\usepackage{algorithm}
\usepackage{algpseudocode}

\usepackage{bbold}
\usepackage[dvips]{color}

\newcommand{\high}[1]{{\color{black}{#1}}}

\begin{document}

\title{\LARGE \bf
Achieving Optimal Throughput and Near-Optimal Asymptotic Delay Performance 
in Multi-Channel Wireless Networks with Low Complexity: A Practical Greedy Scheduling Policy } 


\author{Bo~Ji, Gagan~R.~Gupta, Manu~Sharma, Xiaojun~Lin, and Ness~B.~Shroff
\thanks{B. Ji is with AT\&T Labs. 
M. Sharma is with Qualcomm Technologies Inc, and made 
contribution on this work while he was at Purdue University.
X. Lin is with School of ECE at Purdue University.
N. B. Shroff is with Departments of ECE 
and CSE at the Ohio State University.
Emails: ji.33@osu.edu, gagan.gupta@iitdalumni.com, sharma50@purdue.edu, linx@ecn.purdue.edu, shroff.11@osu.edu.}
\thanks{A preliminary version of this work was presented at the IEEE INFOCOM 2013, Turin, Italy, April, 2013.}
}

\newcounter{theorem}
\newtheorem{theorem}{\it Theorem}

\newcounter{proposition}
\newtheorem{proposition}{\it Proposition}

\newcounter{lemma}
\newtheorem{lemma}{\it Lemma}

\newcounter{corollary}
\newtheorem{corollary}{\it Corollary}

\newcounter{definition}
\newtheorem{definition}{\it Definition}

\newcounter{assumption}
\newtheorem{assumption}{\it Assumption}

\newcommand{\Graph}{\mathcal{G}}
\newcommand{\Node}{\mathcal{V}}
\newcommand{\Vertex}{\mathcal{V}}
\newcommand{\Edge}{\mathcal{E}}
\newcommand{\Link}{\mathcal{L}}
\newcommand{\Flow}{\mathcal{S}}
\newcommand{\Route}{\mathcal{H}}
\newcommand{\bZ}{\mathbb{Z}}
\newcommand{\uZ}{\underline{Z}}
\newcommand{\uP}{\underline{P}}
\newcommand{\uR}{\underline{R}}
\newcommand{\uB}{\underline{B}}
\newcommand{\uU}{\underline{U}}
\newcommand{\uQ}{\underline{Q}}
\newcommand{\utP}{\tilde{\underline{P}}}
\newcommand{\utB}{\tilde{\underline{B}}}
\newcommand{\tB}{\tilde{B}}
\newcommand{\Hop}{{H}}
\newcommand{\Matching}{\mathcal{M}}
\newcommand{\Pair}{\mathcal{P}}
\newcommand{\Int}{\mathbb{Z}}
\newcommand{\Expect}{\mathbf{E}}
\newcommand{\System}{\mathcal{Y}}
\newcommand{\Markov}{\mathcal{X}}
\newcommand{\Prob}{\mathbb{P}}
\newcommand{\vM}{\vec{M}}
\newcommand{\bM}{\mathbf{M}}
\newcommand{\s}{{(s)}}
\newcommand{\sk}{{s,k}}
\newcommand{\hsk}{{\hat{s},\hat{k}}}
\newcommand{\rj}{{r,j}}
\newcommand{\xm}{{x_m}}
\newcommand{\xml}{{x_{m_l}}}
\newcommand{\UGraph}{{U}}
\newcommand{\UVertex}{{X}}
\newcommand{\UEdge}{{Y}}
\newcommand{\vlambda}{\vec{\lambda}}
\newcommand{\vpi}{\vec{\pi}}
\newcommand{\vphi}{\vec{\phi}}
\newcommand{\vpsi}{\vec{\psi}}
\newcommand{\blambda}{\mathbf{\lambda}}
\newcommand{\tM}{\tilde{M}}
\newcommand{\tpi}{\tilde{\pi}}
\newcommand{\tphi}{\tilde{\phi}}
\newcommand{\tpsi}{\tilde{\psi}}
\newcommand{\vmu}{\vec{\mu}}
\newcommand{\vnu}{\vec{\nu}}
\newcommand{\valpha}{\vec{\alpha}}
\newcommand{\vbeta}{\vec{\beta}}
\newcommand{\ve}{\vec{e}}
\newcommand{\vxi}{\vec{\xi}}
\newcommand{\hgamma}{\gamma}
\newcommand{\GMSLambda}{\Lambda_{\text{\it GMS}}}
\newcommand{\argmax}{\operatornamewithlimits{argmax}}
\newcommand{\card}{\aleph}
\newcommand{\V}{\mathcal{V}}
\newcommand{\X}{\mathcal{X}}
\newcommand{\Y}{\mathcal{Y}}
\newcommand{\Z}{\mathcal{Z}}
\newcommand{\mS}{\mathcal{S}}
\newcommand{\mP}{\mathcal{P}}
\newcommand{\mM}{\mathcal{M}}
\newcommand{\mE}{\mathcal{E}}
\newcommand{\mA}{\mathcal{A}}
\newcommand{\mB}{\mathcal{B}}
\newcommand{\mC}{\mathcal{C}}
\newcommand{\mD}{\mathcal{D}}
\newcommand{\indicator}{\mathbb{1}}

\maketitle

\begin{abstract}
In this paper, we focus on the scheduling problem in multi-channel wireless 
networks, e.g., the downlink of a single cell in fourth generation (4G) 
OFDM-based cellular networks. Our goal is to design practical scheduling 
policies that can achieve \emph{provably good performance} in terms of both 
\emph{throughput} and \emph{delay}, at a \emph{low complexity}. While a class
of $O(n^{2.5} \log n)$-complexity hybrid scheduling policies are recently 
developed to guarantee both rate-function delay optimality (in the many-channel 
many-user \emph{asymptotic} regime) and throughput optimality (in the
general \emph{non-asymptotic} setting), their practical complexity 
is typically high. 
To address this issue, we develop a simple greedy policy called \emph{Delay-based 
Server-Side-Greedy (D-SSG)} with a \emph{lower complexity $2n^2+2n$}, 
and rigorously prove that D-SSG not only achieves \emph{throughput optimality}, 
but also guarantees \emph{near-optimal asymptotic delay performance}. 
Specifically, we define the delay-violation probability as the steady-state 
probability that the largest packet waiting time in the system exceeds a certain 
fixed integer threshold $b>0$, and we study the rate-function (or decay-rate) of 
such delay-violation probability when the number of channels or users, $n$, goes 
to infinity. We show that the rate-function attained by D-SSG for any such threshold 
$b$, is no smaller than the maximum achievable rate-function by any scheduling 
policy for threshold $b-1$. Thus, we are able to achieve a reduction in complexity 
(from $O(n^{2.5} \log n)$ of the hybrid policies to $2n^2 + 2n$) with a 
minimal drop in the delay performance. More importantly, in practice, D-SSG generally 
has a substantially lower complexity than the hybrid policies that typically have 
a large constant factor hidden in the $O(\cdot)$ notation. Finally, we conduct 
numerical simulations to validate our theoretical results in various scenarios. 
\high{The simulation results show that in all scenarios we consider, D-SSG not 
only guarantees a near-optimal rate-function, but also empirically has a similar 
delay performance to the rate-function delay-optimal policies.}

\end{abstract}

\section{Introduction} \label{sec:intro}

In this paper, we consider the scheduling problem in a multi-channel 
wireless network, where the system has a large bandwidth that can be 
divided into multiple orthogonal sub-bands (or channels). A practically 
important example of such a multi-channel network is the downlink of 
a single cell of a fourth generation (4G) OFDM-based wireless cellular 
system (e.g., LTE and WiMax). In such a multi-channel system, a key 
challenge is \emph{how to design efficient scheduling policies that 
can simultaneously achieve high throughput and low delay.} This problem 
becomes extremely critical in OFDM systems that are expected to meet 
the dramatically increasing demands from multimedia applications with 
more stringent Quality-of-Service (QoS) requirements (e.g., voice and 
video applications), and thus look for new ways to achieve higher data 
rates, lower latencies, and a much better user experience. Yet, an even 
bigger challenge is \emph{how to design such high-performance scheduling 
policies at a low complexity.} For example, in the OFDM-based LTE systems, 
the \emph{Transmission Time Interval (TTI)}, within which the scheduling 
decisions need to be made, is only one millisecond. 
On the other hand, there are hundreds of orthogonal channels that need to 
be allocated to hundreds of users. Hence, the scheduling decision has to 
be made within a very short scheduling cycle. 

We consider a single-cell multi-channel system consisting of $n$ channels 
and a proportionally large number of users, with intermittent connectivity 
between each user and each channel. We assume that the Base Station (BS) 
maintains a separate First-in First-out (FIFO) queue associated with each 
user, which buffers the packets for the user to download. 
A series of works studied the delay performance of scheduling policies in 
the large-queue asymptotic regime, where the buffer overflow threshold tends 
to infinity (see \cite{ying06,stolyar08,shakkottai08,venkat10} and references 
therein). One potential difficulty of the large-queue asymptotic is that the 
estimates become accurate only when the queue-length or the delay becomes large.
However, for a practical system that aims to serve a large number of users
with more stringent delay requirements (as anticipated in the 4G systems), 
it is more important to ensure small queue-length and small delay \cite{bodas09}. 
Note that even in the wireline networks, there was a similar distinction between 
the large-buffer asymptotic and the many-source 
asymptotic \cite{courcoubetis96,shakkottai01}. It was shown that the many-source 
asymptotic provides sharper estimates of the buffer violation probability when 
the queue-length threshold is not very large. Hence, the delay metric that we 
focus on in this paper is the \emph{decay-rate} (or called the \emph{rate-function} 
in large-deviations theory) of the steady-state probability that the largest 
packet waiting time in the system exceeds a certain fixed threshold when the 
number of users and the number of channels both go to infinity. (See Eq.~(\ref{eq:rf}) 
for the formal definition of rate-function.) We refer to this setting as the 
\emph{many-channel many-user asymptotic regime}. 

A number of recent works have considered a multi-channel system similar
to ours, but looked at delay from different perspectives. A line of works
focused on queue-length-based metrics: average queue length \cite{kittipiyakul09} 
or queue-length rate-function in the many-channel many-user asymptotic regime 
\cite{bodas09,bodas10,bodas11a,bodas11b}. 
In \cite{kittipiyakul09}, the authors focused on minimizing cost functions 
over a finite horizon, which includes minimizing the expected total queue 
length as a special case. The authors showed that their goal can be achieved 
in two special scenarios: 1) a simple two-user system, and 2) systems where 
fractional server allocation is allowed. 
In \cite{bodas09,bodas10,bodas11a,bodas11b}, delay performance is evaluated 
by the queue violation probability, and its associated rate-function, i.e., 
the asymptotic decay-rate of the probability that the largest queue length 
in the system exceeds a fixed threshold in the many-channel many-user
asymptotic regime. Although \cite{bodas09} and \cite{bodas11b} 
proposed scheduling policies that can guarantee both throughput optimality 
and rate-function optimality, 
\high{there are still a number of important dimensions that have space for 
improvement.}
First, although the decay-rate of the queue violation probability may be 
mapped to that of the delay-violation probability when the arrival process 
is deterministic with a constant rate \cite{venkat10}, this is not true in 
general, especially when the arrivals are correlated over time. Further, 
\cite{sharma11,sharma11b,ji13c} have shown through simulations that good 
queue-length performance does not necessarily imply good delay performance. 
Second, their results on rate-function optimality strongly rely on the 
assumptions that the arrival process is \emph{i.i.d. not only across users, 
but also in time}, and that per-user arrival at any time is no greater than 
the largest channel rate. Third, even under this more restricted model, 
the lowest complexity of their proposed rate-function-optimal 
algorithms is $O(n^3)$. For more general models, no algorithm with 
provable rate-function optimality is provided.

Similar to this paper, another line of work \cite{sharma11,sharma11b,ji13a} 
proposed delay-based scheduling policies\footnote{Delay-based 
policies were first introduced in \cite{mekkittikul96} for scheduling problems 
in Input-Queued switches, and were later studied for wireless networks 
\cite{andrews01,shakkottai02,andrews04,eryilmaz05,sadiq09,neely10,ji13c}. 
Please see \cite{ji13c} and references therein for more discussions on the 
history and the recent development of delay-based scheduling policies.} and 
directly focused on the delay performance rather than the queue-length performance. 
The performance of delay is often harder to characterize, because the delay 
in a queueing system typically does not admit a  Markovian representation. 
The problem becomes even harder in a multi-user system with fading channels 
and interference constraints, where the service rate for individual queues 
becomes more unpredictable. In \cite{sharma11,sharma11b}, the authors developed 
a scheduling policy called Delay Weighted Matching (DWM), which maximizes the 
sum of the delay of the scheduled packets in each time-slot. It has been shown 
in \cite{sharma11,sharma11b,ji13a} that DWM is not only throughput-optimal, 
but also rate-function delay-optimal (i.e., maximizing the \emph{delay rate-function}, 
rather than the \emph{queue-length rate-function} as considered in 
\cite{bodas09,bodas10,bodas11a,bodas11b}). 
Moreover, the authors of \cite{sharma11b} used the derived 
rate-function of DWM to develop a simple threshold policy for admission control 
when the number of users scales linearly with the number of channels in the system.
However, DWM incurs a high complexity $O(n^5)$, which renders it 
impractical for modern OFDM systems with many channels and users (e.g., on the 
order of hundreds). In \cite{ji13a}, the authors proposed a class of hybrid 
scheduling policies with a much lower complexity $O(n^{2.5} \log n)$, while 
still guaranteeing both throughput optimality and rate-function delay optimality
(with an additional minor technical assumption). 
However, the practical complexity of the hybrid policies is still high as the 
constant factor hidden in the $O(\cdot)$ notation is typically large due to 
the required two-stage scheduling operations and the operation of computing 
a maximum-weight matching in the first stage. Hence, scheduling policies with 
an even lower (both theoretical and practical) complexity are needed in the 
multi-user multi-channel systems.

This leads to the following natural but important questions: \emph{Can 
we find scheduling policies that have a significantly lower complexity, 
with comparable or only slightly worse performance? How much complexity 
can we reduce, and how much performance do we need to sacrifice?} In 
this paper, we answer these questions positively. Specifically, we develop 
a \emph{low-complexity} greedy policy that achieves both \emph{throughput 
optimality} and \emph{rate-function near-optimality}. 

We summarize our main contributions as follows.

First, we propose a greedy scheduling policy, called \emph{Delay-based 
Server-Side-Greedy (D-SSG)}. D-SSG, in an iterative manner, allocates 
servers one-by-one to serve a connected queue that has the largest 
head-of-line (HOL) delay. We rigorously prove that D-SSG not only 
achieves throughput optimality, but also guarantees a near-optimal 
rate-function. Specifically, the rate-function attained by D-SSG for 
any fixed integer threshold $b>0$, is \emph{no smaller than the maximum 
achievable rate-function by any scheduling policy for threshold $b-1$}. 
We obtain this result by comparing D-SSG with a new \emph{Greedy 
Frame-Based Scheduling (G-FBS)} policy that can exploit a key property 
of D-SSG. We show that G-FBS policy guarantees a near-optimal rate-function, 
and that D-SSG dominates G-FBS in every sample-path. \emph{To the best 
of our knowledge, this is the first work that shows a near-optimal 
rate-function in the above form, and hence we believe that our proof 
technique is of independent interest.} Also, we remark that the gap 
between the near-optimal rate-function attained by D-SSG and the optimal 
rate-function is likely to be quite small. (See Section~\ref{subsec:nodp} 
for detailed discussion.)

D-SSG is a very simple policy and has a \emph{low complexity $2n^2 + 2n$}. 
Note that the queue-length-based counterpart of D-SSG, called Q-SSG, has been 
studied in \cite{bodas10,bodas11a}. However, there the authors were only able 
to prove a positive (queue-length) rate-function for restricted arrival processes 
that are \emph{i.i.d.} not only across users, but also in time. In contrast, 
we show that D-SSG achieves a rate-function that is not only positive but 
also near-optimal, for more general arrival processes. Thus, we are able 
to achieve a reduction in complexity (from $O(n^{2.5} \log n)$ of the hybrid 
policies \cite{ji13a} to $2n^2 + 2n$) with a minimal drop in the delay performance. 
More importantly, the practical complexity of D-SSG is substantially lower than 
that of the hybrid policies since we can precisely bound the constant factor 
in its complexity.



Further, we conduct simulations to validate our analytical results in various 
scenarios. The simulation results show that in all scenarios we consider, 
D-SSG not only guarantees a near-optimal rate-function, but also empirically 
has a similar delay performance to the rate-function delay-optimal policies. 



The remainder of the paper is organized as follows. In Section~\ref{sec:model},
we describe the details of our system model and performance metrics. In
Section~\ref{sec:ub}, we derive an upper bound on the rate-function that 
can be achieved by any scheduling policy. Then, in Section~\ref{sec:dssg}, 
we present our main results on throughput optimality and near-optimal 
rate-function for our proposed low-complexity greedy policy. 
Further, we conduct numerical simulations in Section~\ref{sec:sim}. Finally, 
we make concluding remarks in Section~\ref{sec:con}.

\section{System Model} \label{sec:model}
We consider a discrete-time model for the downlink of a single-cell multi-channel 
wireless network with $n$ orthogonal channels and $n$ users. In each time-slot, a channel 
can be allocated only to one user, but a user can be allocated with multiple channels 
simultaneously. As in \cite{bodas09,bodas10,bodas11a,bodas11b,sharma11,sharma11b,ji13a}, 
for ease of presentation, we assume that the number of users is equal to the number 
of channels. (If the number of users scales linearly with the number of channels, 
the rate-function delay analysis follows similarly. However, an admission control policy 
needs to be carefully designed if the number of users becomes too large \cite{sharma11b}.)
We let $Q_i$ denote the FIFO queue associated with the $i$-th user, and let $S_j$ denote 
the $j$-th server\footnote{Throughout this paper, we use the terms ``user" and ``queue" 
interchangeably, and use the terms ``channel" and ``server" interchangeably.}. We consider 
the \emph{i.i.d.} ON-OFF channel model under which the connectivity between each queue 
and each server changes between ON and OFF from time to time.  
We also assume unit channel capacity, i.e., at most one packet from $Q_i$ can be served 
by $S_j$ when $Q_i$ and $S_j$ are connected. Let $C_{i,j}(t)$ denote the connectivity 
between queue $Q_i$ and server $S_j$ in time-slot $t$. Then, $C_{i,j}(t)$ can be modeled 
as a Bernoulli random variable with a parameter $q \in (0,1)$, i.e.,
\[
C_{i,j}(t) = \left\{
\begin{array}{ll}
1, & \text{with probability}~ q,\\
0, & \text{with probability}~ 1-q.
\end{array}
\right. 
\]
We assume that all the random variables $C_{i,j}(t)$ are \emph{i.i.d.} 
across all the variables $i,j$ and $t$. Such a network can be modeled 
as a multi-queue multi-server system with stochastic connectivity, as 
shown in Fig.~\ref{fig:system}. Further, we assume that the perfect 
channel state information (i.e., whether each channel is ON or OFF for 
each user in each time-slot) is known at the BS. This is a reasonable 
assumption in the downlink scenario of a single cell in a multi-channel 
cellular system with dedicated feedback channels.

As in the previous works \cite{kittipiyakul09,bodas09,bodas10,sharma11,sharma11b,ji13a}, 
the above \emph{i.i.d.} ON-OFF channel model is a simplification, and 
is assumed only for the analytical results. The ON-OFF model is a good 
approximation when the BS transmits at a fixed achievable rate if the 
SINR level is above a certain threshold at the receiver, and does not 
transmit successfully otherwise. The sub-bands being \emph{i.i.d.} is 
a reasonable assumption when the channel width is larger than the 
coherence bandwidth of the environment. Moreover, we believe that our 
results obtained for this channel model can provide useful insights for 
more general models. Indeed, we will show through simulations that our 
proposed greedy policies also perform well in more general models, e.g., 
accounting for heterogeneous (near- and far-)users and time-correlated 
channels. Further, we will briefly discuss how to design efficient 
scheduling policies in general scenarios towards the end of this paper.

\begin{figure}[t]
\centering
\epsfig{file=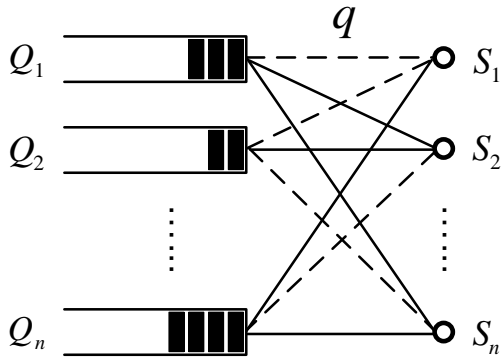,width=0.4\linewidth}
\caption{System model. The connectivity between each pair of queue $Q_i$ and 
server $S_j$ is ``ON" (denoted by a solid line) with probability $q$, and ``OFF" 
(denoted by a dashed line) otherwise.}
\label{fig:system}
\end{figure}

We present more notations used in this paper as follows. Let $A_i(t)$ 
denote the number of packet arrivals to queue $Q_i$ in time-slot $t$. 
Let $A(t)=\sum_{i=1}^n A_i(t)$ denote the cumulative arrivals to the 
entire system in time-slot $t$, and let $A(t_1,t_2) =\sum_{\tau=t_1}^{t_2} A(\tau)$ 
denote the cumulative arrivals to the system from time $t_1$ to $t_2$. 
We let $\lambda_i$ denote the mean arrival rate to queue $Q_i$, and let 
$\lambda \triangleq [\lambda_1,\lambda_2, \dots, \lambda_n]$ denote the 
arrival rate vector. We assume that packets arrive at the beginning of a 
time-slot, and depart at the end of a time-slot. We use $Q_i(t)$ to denote 
the length of queue $Q_i$ at the beginning of time-slot $t$ immediately 
after packet arrivals. Queues are assumed to have an infinite buffer 
capacity. Let $Z_{i,l}(t)$ denote the delay (or waiting time) of the 
$l$-th packet at queue $Q_i$ at the beginning of time-slot $t$, \emph{which 
is measured from the time when the packet arrived to queue $Q_i$ until the 
beginning of time-slot $t$}. Note that at the end of each time-slot, the 
packets that are still present in the system will have their delays increased 
by one due to the elapsed time. Further, let $W_i(t)=Z_{i,1}(t)$ (or $W_i(t)=0$
if $Q_i(t)=0$) denote the HOL delay of queue $Q_i$ at the beginning of time-slot 
$t$. Finally, we define $(x)^+ \triangleq \max(x,0)$, and use $\mathbb{1}_{\{ \cdot \}}$ 
to denote the indicator function.

We now state the assumptions on the arrival processes. The throughput 
analysis is carried out under Assumption~\ref{ass:arr_slln} only, which 
is mild and has also been used in \cite{andrews04,ji13a}.

\begin{assumption}
\label{ass:arr_slln}
For each user $i \in \{1,2,\dots,n \}$, the arrival process $A_i(t)$ 
is an irreducible and positive recurrent Markov chain with countable 
state space, and satisfies the Strong Law of Large Numbers: That is, 
with probability one,
\begin{equation}
\label{eq:slln}
\lim_{t \rightarrow \infty} \frac {\sum^{t-1}_{\tau=0} A_i(\tau)} {t} = \lambda_i.
\end{equation}
We also assume that the arrival processes are mutually independent across users 
(which can be relaxed for throughput analysis as discussed in \cite{andrews04}).
\end{assumption}

The rate-function delay analysis is carried out under the following 
two assumptions, which have also been used in the previous works 
\cite{sharma11,sharma11b,ji13a}.

\begin{assumption}
\label{ass:arr_bound}
There exists a finite $L$ such that $A_i(t) \le L$ for any $i$ and $t$,
i.e., instantaneous arrivals are bounded.
\end{assumption}

\begin{assumption}
\label{ass:arr_ld}
The arrival processes are \emph{i.i.d.} across users, and $\lambda_i=p$ for any 
user $i$. Given any $\epsilon>0$ and $\delta>0$, there exist $T>0$, $N>0$, and 
a positive function $I_B(\epsilon,\delta)$ independent of $n$ and $t$ such that
\[
\Prob ( \frac {\sum_{\tau=1}^{t} \mathbb{1}_{\{|\sum_{i=1}^n A_i(\tau) - pn|
> \epsilon n \}}} {t} > \delta ) < \exp (-nt I_B(\epsilon,\delta)),
\]
for all $t>T$ and $n>N$. 
\end{assumption}

Assumption~\ref{ass:arr_bound} requires that the arrivals in each time-slot have 
bounded support, which is indeed true for practical systems. Assumption~\ref{ass:arr_ld} 
is also very general, and can be viewed as a result of the statistical multiplexing 
effect of a large number of sources. Assumption~\ref{ass:arr_ld} holds for \emph{i.i.d.} 
arrivals and arrivals driven by two-state Markov chains (that can be correlated over 
time) as two special cases (see Lemmas 2 and 3 of \cite{sharma11b}).



\subsection{Performance Objectives}
In this paper, we consider two performance metrics: 1) the \emph{throughput} and 
2) the steady-state probability that the largest packet delay in the system exceeds 
a certain fixed threshold, and its associated \emph{rate-function} in the many-channel 
many-user asymptotic regime. 

We first define the \emph{optimal throughput region} (or \emph{stability region}) 
of the system for any fixed integer $n>0$ under Assumption~\ref{ass:arr_slln}.
As in \cite{andrews04}, a stochastic queueing network is said to be \emph{stable}
if it can be described as a discrete-time countable Markov chain and the Markov 
chain is stable in the following sense: The set of positive recurrent states is 
nonempty, and it contains a finite subset such that with probability one, this 
subset is reached within finite time from any initial state. When all the states 
communicate, stability is equivalent to the Markov chain being positive recurrent
\cite{bramson08}.
The \emph{throughput region} of a scheduling policy is defined as the set of arrival 
rate vectors for which the network remains stable under this policy. Then, the 
\emph{optimal throughput region} is defined as the union of the throughput regions 
of all possible scheduling policies, which is denoted by $\Lambda^*$. 
A scheduling policy is \emph{throughput-optimal}, if it can stabilize any arrival 
rate vector strictly inside $\Lambda^*$. 
For more discussions on the optimal throughput region $\Lambda^{*}$ in our 
multi-channel systems, please refer to \cite{ji13a}.

Next, we consider the steady-state probability that the largest packet delay 
in the system exceeds a certain fixed threshold, and its associated
\emph{rate-function} in the many-channel many-user asymptotic regime. 
\high{Assuming that the system is stationary and ergodic, let $W(0) \triangleq 
\max_{1 \le i \le n} W_i(0)$ denote the largest HOL delay over all the queues 
(i.e., the largest packet delay in the system) in the steady state,} and then 
we define rate-function $I(b)$ as the decay-rate of the probability that $W(0)$ 
exceeds any fixed integer threshold $b \ge 0$, as the system size $n$ goes to 
infinity, i.e.,
\begin{equation}
\label{eq:rf}
I(b) \triangleq \lim_{n \rightarrow \infty} \frac {-1} {n} \log 
\Prob (W(0)>b).
\end{equation}
Note that once we know this rate-function, we can then estimate the delay-violation
probability using $\Prob (W(0)>b) \approx \exp (-nI(b))$. The estimate tends to be 
more accurate as $n$ becomes larger. Clearly, for systems with a large $n$, a larger 
value of the rate-function implies a better delay performance, i.e., a smaller 
probability that the largest packet delay in the system exceeds a certain threshold. 
As in \cite{sharma11,sharma11b,ji13a}, 
we define the \emph{optimal rate-function} as the maximum achievable rate-function 
over all possible scheduling policies, which is denoted by $I^*(b)$. A scheduling 
policy is \emph{rate-function delay-optimal} if it achieves the optimal rate-function 
$I^*(b)$ for any fixed integer threshold $b \ge 0$.



\section{An Upper Bound on The Rate-Function} \label{sec:ub}
In this section, we derive an upper bound of the rate-function for all 
scheduling policies.

Let $I_{AG}(t,x)$ denote the asymptotic decay-rate of the probability 
that in any interval of $t$ time-slots, the total number of arrivals 
is greater than $n(t+x)$, as $n$ tends to infinity, i.e.,  
\[
I_{AG}(t,x) \triangleq \liminf_{n \rightarrow \infty} \frac {-1}{n} 
\log \Prob(A(-t+1,0) > n(t+x)).
\]
Let $I_{AG}(x)$ be the infimum of $I_{AG}(t,x)$ over all $t>0$, i.e., 
\[
I_{AG}(x) \triangleq \inf_{t>0} I_{AG}(t,x).
\]
Also, we define $I_X \triangleq \log \frac {1} {1-q}$.

\begin{theorem}
\label{thm:ub}
Given the system model described in Section~\ref{sec:model},
for any scheduling algorithm, we have
\[
\begin{split}
& \limsup_{n \rightarrow \infty} \frac {-1}{n} \log \Prob(W(0)>b) \\
& \le \min \{(b+1)I_X, \min_{0 \le c \le b} \{ I_{AG}(b-c) + c I_X \} \} \triangleq I_U(b).
\end{split}
\]
\end{theorem}

Theorem~\ref{thm:ub} can be shown by considering two types of events that 
lead to the delay-violation event $\{W(0)>b\}$ no matter how packets are scheduled, 
and computing their probabilities and decay-rates. In the above expression of $I_U(b)$, 
the first term $(b+1)I_X$ is due to sluggish services, which corresponds to the event 
that a queue with at least one packet is disconnected from all of the $n$ servers for 
consecutive $b+1$ time-slots. The second term $\min_{0 \le c \le b} \{ I_{AG}(b-c) + c I_X \}$ 
is due to both bursty arrivals and sluggish services, where 
$I_{AG}(b-c)$ corresponds to the event that the arrivals are too bursty during the 
interval of $[-t-b, -b-1]$ such that at the beginning of time-slot $-c$ for $c \le b$, 
there exists at least one packet remaining in the system, say queue $Q_1$. Then, the 
term $c I_X$ corresponds to the event that the services are too sluggish such that 
queue $Q_1$ is disconnected from all of the $n$ servers for the following consecutive 
$c$ time-slots. Clearly, both of the above events will lead to the delay-violation 
event $\{W(0)>b\}$ under all scheduling policies. 
We provide the detailed proof in Appendix~\ref{app:thm:ub}.

\emph{Remark:} Theorem~\ref{thm:ub} implies that $I_U(b)$ is an upper bound 
on the rate-function that can be achieved by any scheduling policy. Hence, 
even for the optimal rate-function $I^*(b)$, we must have $I^*(b) \le I_U(b)$ 
for any fixed integer threshold $b \ge 0$. 

Note that our derived upper bound $I_U(b)$ is strictly positive in the cases of interest. 
For example, when $L=1$, it has been shown in \cite{ji13a} that 
the optimal rate-function is $I^*(b) = (b+1) \log \frac {1} {1-q}$, and thus $I_U(b) 
\ge I^*(b) > 0$ for all integer $b \ge 0$. This holds for general arrival processes 
under Assumptions~\ref{ass:arr_bound} and \ref{ass:arr_ld}, including two special 
cases of \emph{i.i.d.} Bernoulli arrivals and two-state Markov chain driven arrivals. 
When $L > 1$, in Appendix~\ref{app:posi_rf} we show that $I_U(b)$ is strictly positive 
for the special case of \emph{i.i.d.} 0-$L$ arrivals with feasible arrival rates;
further, in Section~\ref{sec:sim}, our simulation results (Fig.~\ref{fig:mc}) also 
demonstrate that the rate-function attained by D-SSG is strictly positive under 
two-state Markov chain driven arrivals.
\section{Delay-based Server-Side-Greedy (D-SSG)} \label{sec:dssg}
In \cite{ji13a}, it has been shown that a class of two-stage hybrid policies 
can achieve both throughput optimality and rate-function delay optimality at 
a lower complexity $O(n^{2.5} \log n)$ (compared to $O(n^5)$ of DWM). The hybrid 
policies are constructed by combining certain throughput-optimal policies with 
a rate-function delay-optimal policy DWM-$n$ (where $n$ is the number 
of users or channels), which in each time-slot maximizes the sum of the delay 
of the scheduled packets among the $n$ oldest packets in the system. For example, 
DWM-$n$ combined with the Delay-based MaxWeight Scheduling (D-MWS) policy 
\cite{ji13a,andrews04,eryilmaz05} yields a $O(n^{2.5} \log n)$ complexity hybrid 
policy, called the DWM-$n$-MWS policy. 

The above result leads to the following important questions: \emph{Is it possible 
to develop scheduling policies with an even lower complexity, while achieving 
comparable or only slightly worse performance? If so, how much complexity 
can we reduce, and how much performance do we need to sacrifice?} In this 
section, we answer these questions positively. We first develop a greedy 
scheduling policy called \emph{Delay-based Server-Side-Greedy (D-SSG)} with 
an even lower complexity $2n^2+2n$. Under D-SSG, each server iteratively chooses 
to serve a connected queue that has the largest HOL delay. Then, we show that 
D-SSG not only achieves throughput optimality, but also guarantees a near-optimal 
rate-function. Hence, D-SSG achieves a reduction in complexity (from $O(n^{2.5} 
\log n)$ of the hybrid policies to $2n^2+2n$) with a minimal drop in the delay 
performance. More importantly, the practical complexity of D-SSG is substantially 
lower than that of the hybrid policies.

\subsection{Algorithm Description}
Before we describe the detailed operations of D-SSG, we would like to remark 
on the D-MWS policy in our multi-channel system, due to the similarity between 
D-MWS and D-SSG. Under D-MWS, each server chooses to serve a queue that has 
the largest HOL delay (among all the queues connected to this server). Note 
that D-MWS is not only throughput-optimal, but also has a low complexity 
$O(n^2)$. However, in \cite{ji13a} it has been shown that D-MWS suffers from 
poor delay performance. (Specifically, D-MWS yields a rate-function of zero 
in certain scenarios, e.g., with \emph{i.i.d.} 0-1 arrivals). The reason 
is that under D-MWS, each server chooses to serve a connected queue that 
has the largest HOL delay without accounting for the decisions of the other 
servers. This way of allocating servers leads to an unbalanced schedule. 
That is, only a small fraction of the queues get served in each time-slot. 
This inefficiency leads to poor delay performance.




Now, we describe the operations of our proposed D-SSG policy. D-SSG is similar 
to D-MWS, in the sense that it also allocates each server to a connected queue 
that has the largest HOL delay. However, the key difference is that, instead 
of allocating the servers all at once as in D-MWS, D-SSG allocates the servers 
one-by-one, accounting for the scheduling decisions of the servers that are 
allocated earlier. We will show that \emph{this critical difference results in 
a substantial improvement in the delay performance.} 

We present some additional notations, and then specify the detailed operations 
of D-SSG. In each time-slot, there are $n$ rounds, and in each round, one of 
the remaining servers is allocated. Let $Q_i^k(t)$, $Z^k_{i,l}(t)$ and 
$W_i^k(t)=Z^k_{i,1}(t)$ (or $W_i^k(t)=0$ if $Q_i^k(t)=0$) denote the length of 
queue $Q_i$, the delay of the $l$-th packet of $Q_i$, and the HOL delay of $Q_i$ 
after $k \ge 1$ rounds of server allocation in time-slot $t$, respectively. In 
particular, we have $Q_i^0(t)=Q_i(t)$, $Z^0_{i,l}(t)=Z_{i,l}(t)$, and $W_i^0(t)=W_i(t)$. 
Let $\mS_j(t)$ denote the set of queues being connected to server $S_j$ in time-slot 
$t$, i.e., $\mS_j(t) = \{1 \le i \le n ~|~ C_{i,j}(t)=1 \}$. Let $\Gamma^k_j(t)$ 
denote the set of indices of the queues that are connected to server $S_j$ in 
time-slot $t$ and that have the largest HOL delay at the beginning of the $k$-th 
round in time-slot $t$, i.e., $\Gamma^k_j(t) \triangleq \{i \in \mS_j(t) ~|~ 
W_i^{k-1}(t) = \max_{l \in \mS_j(t)} W_l^{k-1}(t) \}$. Let $i(j,t)$ denote the 
index of queue that is served by server $S_j$ in time-slot $t$ under D-SSG. 
 
\noindent {\bf Delay-based Server-Side-Greedy (D-SSG) policy:}
In each time-slot $t$,

\begin{enumerate}
\item Initialize $k=1$.
\item In the $k$-th round, allocate server $S_k$ to serve queue $Q_{i(k,t)}$,
where $i(k,t) = \min \{ i ~|~ i \in  \Gamma^k_k(t)\}$. That is, in the $k$-th 
round, the $k$-th server $S_k$ is allocated to serve the connected queue that 
has the largest HOL delay, breaking ties by picking the queue with the smallest 
index if there are multiple such queues. Then, update the length of $Q_{i(k,t)}$ 
to account for service, i.e., set $Q_{i(k,t)}^k(t) = \left(Q_{i(k,t)}^{k-1}(t) 
- C_{i(k,t),k}(t) \right)^+$ and $Q_i^k(t) = Q_i^{k-1}(t)$ for all $i \neq i(k,t)$. 
Also, update the HOL delay of $Q_{i(k,t)}$ to account for service, i.e., set 
$W_{i(k,t)}^k(t) = Z^k_{i(k,t),1}(t) = Z^{k-1}_{i(k,t),2}(t)$ if $Q_{i(k,t)}^k(t)>0$, 
and $W_{i(k,t)}^k(t)=0$ otherwise, and set $W_i^k(t) = W_i^{k-1}(t)$ for all 
$i \neq i(k,t)$.

\item Stop if $k$ equals $n$. Otherwise, increase $k$ by 1 and repeat step 2.
\end{enumerate} 

\emph{Remark:} 
From the above operations, it can be observed that in each round, D-SSG aims 
to allocate the available server with the smallest index. Further, when there 
are multiple queues that are connected to the considered server and that have 
the largest HOL delay, D-SSG favors the queue with the smallest index. We specify 
such tie-breaking rules for ease of analysis only. In practice, we can break ties 
arbitrarily.

We highlight that D-SSG has a low complexity of $2n^2+2n$ due to the following 
operations. Assume that each packet contains the information of its arriving 
time. At the beginning of each time-slot, it requires $n$ addition operations 
to update the HOL delay of each of the $n$ queues (i.e., increasing it by one). 
In each round $k$, it takes $n$ time to check the connectivity between server 
$S_k$ and the $n$ queues, another up to $n$ time to find the connected queue 
with the largest HOL delay, and one more basic operation to update the HOL delay 
of the queue chosen by server $S_k$. Since there are $n$ rounds, the overall 
complexity is $n+n(n+n+1)=2n^2+2n$.

Note that the queue-length-based counterpart of D-SSG, called Q-SSG, has been 
studied in \cite{bodas10,bodas11a}. Under Q-SSG, each server iteratively chooses 
to serve a connected queue that has the largest length. It has been shown that 
Q-SSG not only achieves throughput optimality, but also guarantees a \emph{positive 
(queue-length)} rate-function. However, their results have the following limitations: 
1) a positive rate-function may not be good enough, since the gap between the 
guaranteed rate-function and the optimal is unclear; 2) good queue-length performance 
does not necessarily translate into good delay performance; 3) their analysis 
was only carried out for restricted arrival processes that are \emph{not only 
i.i.d. across users, but also in time}. In contrast, in this section we will show 
that D-SSG achieves a rate-function that is \emph{not only positive but also 
near-optimal} (in the sense of (\ref{eq:norf})) for more general arrival processes, 
while guaranteeing throughput optimality.

\subsection{Throughput Optimality} \label{subsec:to}
We first establish throughput optimality of D-SSG in general non-asymptotic
settings with any fixed value of $n$. Note that in Section~\ref{subsec:nodp}, 
we will analyze the delay performance of D-SSG in the asymptotic regime, 
where $n$ goes to infinity. Hence, even if the convergence rate of the delay 
rate-function is fast (as is typically the case), the throughput performance 
may still be poor for small to moderate values of $n$. As a matter of fact, 
for a fixed $n$, a rate-function delay-optimal policy (e.g., DWM-$n$) may not 
even be throughput-optimal \cite{ji13a}. To this end, we first focus on studying 
the throughput performance of D-SSG in general non-asymptotic settings. 

We remark that the throughput performance of scheduling policies have been
extensively studied in various settings, including the multi-channel systems
that we consider in this paper. Specifically, for such multi-channel systems, 
\cite{ji13a} proposed a class of Maximum Weight in the Fluid limit (MWF) 
policies and proved throughput-optimality of the MWF policies in very general 
settings (under Assumption~\ref{ass:arr_slln}). The key insight is that to 
achieve throughput-optimality in such multi-channel systems, it is sufficient 
for each server to choose a connected queue with a large enough weight (i.e.,
queue-length or delay) such that this queue has the largest weight in the fluid 
limit \cite{dai95}. 

Next, we prove that D-SSG is throughput-optimal in general non-asymptotic 
settings (for a system with any fixed value of $n$) by showing that D-SSG 
is an MWF policy.

\begin{theorem}
\label{thm:dssg-to}
D-SSG policy is throughput-optimal under Assumption~\ref{ass:arr_slln}.
\end{theorem}

We provide the detailed proof in Appendix~\ref{app:thm:dssg-to}.

\subsection{Near-optimal Asymptotic Delay Performance} \label{subsec:nodp}
In this subsection, we present our main result on the near-optimal rate-function. 
We first define near-optimal rate-function, and then evaluate the delay performance 
of D-SSG. 

A policy $\mathbf{P}$ is said to achieve \emph{near-optimal rate-function} 
if the delay rate-function $I(b)$ attained by policy $\mathbf{P}$ for any 
fixed integer threshold $b>0$, is no smaller than $I^*(b-1)$, the optimal 
rate-function for threshold $b-1$. That is,
\begin{equation}
\label{eq:norf}
I(b)=\liminf_{n \rightarrow \infty} \frac {-1} {n} \log \Prob \left( W(0) > b \right) 
\ge I^*(b-1).
\end{equation}

We next present our main result of this paper in the following theorem, 
which states that D-SSG achieves a near-optimal rate-function.
 
\begin{theorem}
\label{thm:dssg-lb}
Under Assumptions~\ref{ass:arr_bound} and \ref{ass:arr_ld}, D-SSG
achieves a near-optimal rate-function, as given in Eq. (\ref{eq:norf}).
\end{theorem}

\high{
\emph{We prove Theorem~\ref{thm:dssg-lb} by the following strategy:
1) motivated by a key property of D-SSG (Lemma~\ref{lem:success}), we 
propose the \emph{Greedy Frame-Based Scheduling (G-FBS)} policy, which 
is a variant of the FBS policy \cite{sharma11,sharma11b} that has been 
shown to be rate-function delay-optimal in some cases;
2) show that G-FBS achieves a near-optimal rate-function (Theorem~\ref{thm:g-fbs}); 
3) prove a dominance property of D-SSG over G-FBS. Specifically, in 
Lemma~\ref{lem:dssg-dom}, we show that for any given sample path, by 
the end of each time-slot, D-SSG has served every packet that G-FBS 
has served.} 
}

We now present a crucial property of D-SSG in Lemma~\ref{lem:success}, 
which is the key to proving a near-optimal rate-function for D-SSG.

\high{
\begin{lemma}
\label{lem:success}
Consider a set of $n$ packets satisfying that no more than $2H$ packets 
are from the same queue, where $H>4$ is any integer constant independent
of $n$. Consider any strictly increasing function $f(n)$ such that $f(n) 
< \frac {n} {2}$ and $f(n) \in o(n/\log^2 n)$. Suppose that D-SSG is 
applied to schedule these $n$ packets. Then, there exists a finite integer 
$N_X>0$ such that for all $n \ge N_X$, with probability no smaller than 
$1-2(1-q)^{n-f(n)\log^2n}$, D-SSG schedules at least $n-H\sqrt{n}$ packets, 
including the oldest $f(n)$ packets among the $n$ packets.
\end{lemma}

To prove Lemma~\ref{lem:success} and thus near-optimal rate-function of 
D-SSG (Theorem~\ref{thm:dssg-lb}), we introduce another greedy scheduling 
policy called \emph{Delay-based Queue-Side-Greedy (D-QSG)} and a sample-path 
equivalence property between D-QSG and D-SSG (Lemma~\ref{lem:equivalent}). 
Please refer to Appendix~\ref{app:sec:dqsg} for details.
}

We provide the proof of Lemma~\ref{lem:success} in Appendix~\ref{app:lem:success},
and explain the importance of Lemma~\ref{lem:success} as follows. 
We first recall how DWM is shown to be rate-function delay-optimal 
(for some cases) in \cite{sharma11,sharma11b}. Specifically, the 
authors of \cite{sharma11,sharma11b} compare DWM with another policy 
FBS. In FBS, packets are filled into frames with size $n-H$ in a 
First-Come First-Serve (FCFS) manner such that no two packets in the 
same frame have a delay difference larger than $h$ time-slots, where 
$h>0$ is a suitably chosen constant independent of $n$ and $H=Lh$. 
The FBS policy attempts to serve the entire HOL frame whenever possible. 
The authors of \cite{sharma11,sharma11b} first establish the rate-function 
optimality of the FBS policy. Then, by showing that DWM dominates FBS 
(i.e., DWM will serve the same packets in the entire HOL frame whenever 
possible), the delay optimality of DWM then follows.

However, this comparison approach will not work directly for D-SSG. 
In order to serve all packets in a frame whenever possible, one would 
need certain back-tracking (or rematching) operations as in a typical 
maximum-weight matching algorithm like DWM. For a simple greedy algorithm 
like D-SSG that does not do back-tracking, it is unlikely to attain the 
same probability of serving the entire frame. In fact, even if we reduce 
the maximum frame size to $n-H\sqrt{n}$, we are still unable to show that 
D-SSG can serve the entire frame with a sufficiently high probability. 
Thus, we cannot compare D-SSG with FBS as in \cite{sharma11,sharma11b}. 

Fortunately, Lemma~\ref{lem:success} provides an alternate avenue. 
Specifically, for a set of $n$ packets, even though D-SSG may not 
serve any \emph{given} subset of $n-H\sqrt{n}$ packets with a 
sufficiently high probability, it will serve \emph{some} subset of 
$n-H\sqrt{n}$ packets with a sufficiently high probability. Further, 
this subset must contain the oldest $H\sqrt{n}$ packets for a large
$n$, if we choose $f(n)$ in Lemma~\ref{lem:success} such that $f(n) 
> H\sqrt{n}$ for large $n$. Note that D-SSG still leaves (at most) 
$H\sqrt{n}$ packets to the next time-slot. If we can ensure that in 
the next time-slot, D-SSG serves all of these $H\sqrt{n}$ leftover 
packets, we would then at worst suffer an additional one-time-slot 
delay. Indeed, Lemma~\ref{lem:success} guarantees this with high 
probability. Intuitively, we would then be able to show that D-SSG 
attains a near-optimal delay rate-function as given in Eq. (\ref{eq:norf}). 

To make this argument rigorous, we next compare D-SSG with a new policy 
called \textbf{Greedy Frame-Based Scheduling (G-FBS)}. Note that G-FBS 
is only for assisting our analysis, and will not be used as an actual 
scheduling algorithm. We first fix a properly chosen parameter $h>0$.
In the G-FBS policy, packets are grouped into frames satisfying the
following requirements: 1) No two packets in the same frame have a 
delay difference larger than $h$ time-slots. This guarantees that
in a frame, no more than $H=Lh$ packets from the same queue can be
filled into a single frame; 2) Each frame has a capacity of 
$n_0=n-H\sqrt{n}$ packets, i.e., at most $n_0$ packets can be filled 
into a frame; 3) As packets arrive to the system in each time-slot, 
the frames are created by filling the packets sequentially. Specifically, 
packets that arrive earlier are filled into the frame with a higher 
priority, and packets from queues with a smaller index are filled with 
a higher priority when multiple packets arrive in the same time-slot. 
Once any of the above requirements is violated, the current frame will 
be closed and a new frame will be open. We also assume that there is a 
``leftover" frame, called \emph{L-frame} for simplicity, with a capacity 
of $H\sqrt{n}$ packets. The L-frame is for storing the packets that were 
not served in the previous time-slot and were carried over to the current 
time-slot. At the beginning of each time-slot, we combine the HOL frame 
and the L-frame into a ``super" frame, called \emph{S-frame} for simplicity, 
with a capacity of $n$ packets. It is easy to see that in the S-frame, no 
more than $2H$ packets are from the same queue. Note that if there are less 
than $n$ packets in the S-frame, we can artificially add some dummy packets 
with a delay of zero at the end of the S-frame so that the S-frame is fully 
filled, but still need to guarantee that no more than $2H$ packets from the 
same queue can be filled into the S-frame. In each time-slot, G-FBS runs 
the D-SSG policy, but restricted to only the 
$n$ packets of the S-frame. We call it a \emph{success}, if D-SSG can 
schedule at least $n_0$ packets, including the oldest $f(n)$ packets, 
from the S-frame, where $f(n) < \frac {n}{2}$ is any function that 
satisfies that $f(n) \in o(n/\log^2 n)$ and $f(n) \in \omega(\sqrt{n})$. In each 
time-slot, if a success does not occur, then no packets will be served. 
When there is a success, the G-FBS policy serves all the packets that 
are scheduled by D-SSG restricted to the S-frame in that time-slot. 
\emph{Lemma~\ref{lem:success} implies that in each time-slot, a success 
occurs with probability at least $1-2(1-q)^{n-f(n) \log^2 n}$.} A success 
serves all packets from the S-frame, except for at most $H\sqrt{n}=n-n_0$ 
packets, and these served packets include the oldest $f(n)$ packets. The 
packets that are not served will be stored in the L-frame, and carried 
over to the next time-slot (except for the dummy packets, which will be 
discarded). 

\emph{Remark:} 
Although G-FBS is similar to FBS policy \cite{sharma11,sharma11b}, it exhibits 
a key difference from FBS. In the FBS policy, in each time-slot, either an entire 
frame (i.e., all the packets in the frame) will be completely served or none of 
its packets will be served. Hence, it does not allow packets to be carried over 
to the next time-slot. In contrast, G-FBS allows leftover packets and is thus 
more flexible in serving frames. This property is the key reason that we can use 
lower-complexity policies like D-SSG. On the other hand, it leads to a small gap 
between the rate-functions achieved by G-FBS and delay-optimal policies (e.g., 
DWM and the hybrid policies). Nonetheless, this gap can be well characterized. 
Specifically, in the G-FBS policy, an L-frame contains at most $H\sqrt{n}$ packets 
that are not served whenever there is a success. 
Further, these (at most) $H\sqrt{n}$ leftover packets will be among the oldest 
$f(n)$ packets (in the S-frame) in the next time-slot for large $n$, due to our 
choice of $f(n) \in \omega (\sqrt{n})$. Hence, another success will serve all the 
leftover packets. This implies that at most $x+1$ successes are needed to completely 
serve $x$ frames, for any finite integer $x>0$. \emph{In fact, this property is the 
key reason for a one-time-slot shift in the guaranteed rate-function by G-FBS, which 
leads to the near-optimal delay rate-function}, as we show in the following theorem.


\begin{theorem}
\label{thm:g-fbs}
Under Assumptions~\ref{ass:arr_bound} and \ref{ass:arr_ld}, G-FBS policy 
achieves a near-optimal rate-function, as given in Eq. (\ref{eq:norf}).
\end{theorem}

The proof of Theorem~\ref{thm:g-fbs} follows a similar line of argument 
as in the proof for rate-function delay optimality of FBS (Theorem~2 in 
\cite{sharma11b}). We consider all the events that lead to the delay-violation
event $\{W(0)>b\}$, which can be caused by two factors: bursty arrivals 
and sluggish service. On the one hand, if there are a large number of 
arrivals in certain period, say of length $t$ time-slots, which exceeds 
the maximum number of packets that can be served in a period of $t+b+1$ 
time-slots, then it unavoidably leads to a delay-violation. On the other 
hand, suppose that there is at least one packet arrival at certain time, 
and that under G-FBS, a success does not occur in any of the following 
$b+1$ time-slots (including the time-slot when the packet arrives), then 
it also leads to a delay-violation. Each of these two possibilities has 
a corresponding rate-function for its probability of occurring. Large-deviations 
theory then tells us that the rate-function for delay-violation is 
determined by the smallest rate-function among these possibilities (i.e., 
``rare events occur in the most likely way"). We can then show that $I(b) 
\ge I_U(b-1) \ge I^*(b-1)$ for any integer $b>0$, where $I(\cdot)$ is the 
rate-function attained by G-FBS, $I_U(\cdot)$ is the upper bound that we 
derived in Section~\ref{sec:ub}, and $I^*(\cdot)$ is the optimal rate-function, 
respectively. 
We provide the detailed proof of Theorem~\ref{thm:g-fbs} in Appendix~\ref{app:thm:g-fbs}.

\emph{Remark:}
Note that the gap between the optimal rate-function and the above near-optimal 
rate-function is likely to be quite small. For example, in the case where the  
arrival is either 1 or 0, the near-optimal rate-function implies $I(b) \ge 
\frac{b}{b+1} I^*(b)$, since we have $I^*(b) = (b+1) \log \frac{1}{1-q}$ for 
this case \cite{ji13a}.

Finally, we make use of the following dominance property of D-SSG over G-FBS. 

\begin{lemma}
\label{lem:dssg-dom}
For any given sample path and for any value of $h$, by the end of any time-slot
$t$, D-SSG has served every packet that G-FBS has served.
\end{lemma}

We prove Lemma~\ref{lem:dssg-dom} by contradiction. The proof follows a 
similar argument as in the proof of Lemma~7 in \cite{sharma11b}, and is 
provided in Appendix~\ref{app:lem:dssg-dom}.
Then, the near-optimal rate-function of D-SSG (Theorem~\ref{thm:dssg-lb}) 
follows immediately from Lemma~\ref{lem:dssg-dom} and Theorem~\ref{thm:g-fbs}. 

\emph{Remark:} Note that D-SSG combined with DWM-$n$ policy, can 
also yield an $O(n^{2.5} \log n)$-complexity hybrid policy that 
is both throughput-optimal and rate-function delay-optimal. We
omit the details since the treatment follows similarly as that 
for hybrid DWM-$n$-MWS policy \cite{ji13a}.

So far, we have shown that our proposed low-complexity D-SSG policy achieves 
both throughput optimality and near-optimal delay rate-function. 
In the next section, we will show through simulations that \high{in all scenarios 
we consider, \emph{D-SSG not only exhibits a near-optimal delay rate-function, 
but also empirically has a similar delay performance to the rate-function 
delay-optimal policies such as DWM and the hybrid DWM-$n$-MWS policy.}}

\section{Simulation Results} \label{sec:sim}
In this section, we conduct simulations to compare scheduling performance 
of our proposed D-SSG policy with DWM, hybrid DWM-$n$-MWS (called Hybrid 
for short), D-MWS, and Q-SSG. We simulate these policies in Java and compare 
the empirical probabilities that the largest HOL delay in the system in any 
given time-slot exceeds an integer threshold $b$, i.e., $\Prob (W(0)>b)$.


Same as in \cite{ji13a}, we consider bursty arrivals that are driven by 
a two-state Markov chain and that are correlated over time. (We obtained 
similar results for \emph{i.i.d.} 0-$L$ arrivals, and omit them here.) 
For each user, there are 5 packet-arrivals when the Markov chain is in 
state 1, and there is no arrivals when it is in state 2. The transition 
probability of the Markov chain is given by the matrix $[0.5,0.5;0.1,0.9]$,
and the state transitions occur at the end of each time-slot. The arrivals 
for each user are correlated over time, but they are independent across users. 
For the channel model, we first assume \emph{i.i.d.} ON-OFF channels with unit 
capacity, and set $q=0.75$. We later consider more general scenarios with 
heterogeneous users and bursty channels that are correlated over time. We 
run simulations for a system with $n$ servers and $n$ users, where $n \in 
\{ 10,20,\dots,100 \}$. The simulation period lasts for $10^7$ time-slots 
for each policy and each system.

\begin{figure}[t]
\centering
\epsfig{file=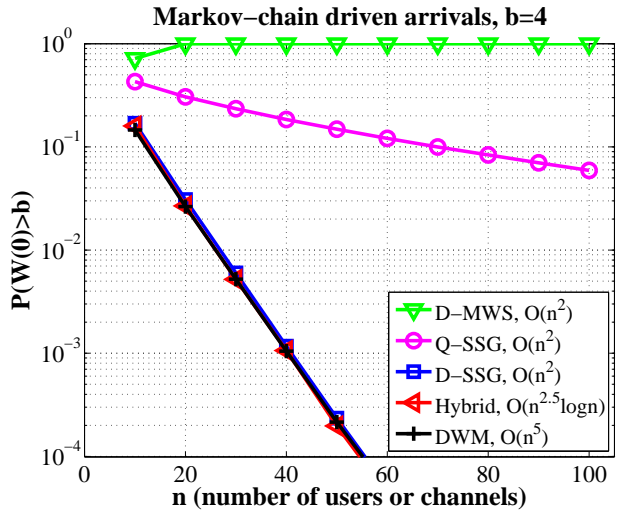,width=0.5\linewidth}
\caption{Performance comparison of different scheduling policies 
in the case with homogeneous \emph{i.i.d.} channels, for $b=4$.}
\label{fig:mc}
\end{figure}

\begin{figure}[t]
\centering
\epsfig{file=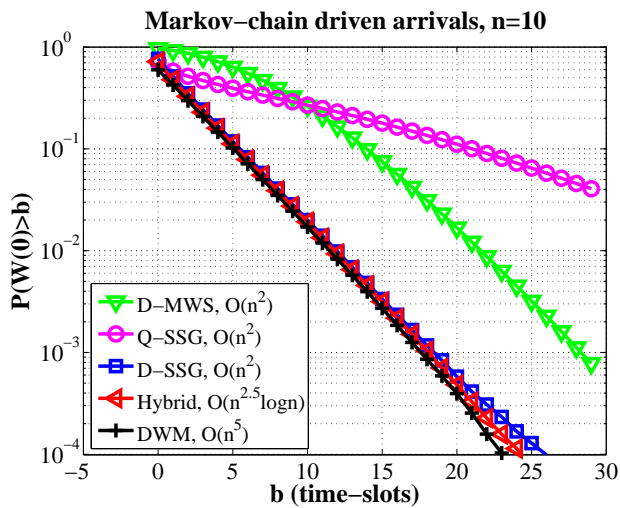,width=0.5\linewidth}
\caption{Performance comparison of different scheduling policies 
in the case with homogeneous \emph{i.i.d.} channels, for $n=10$.}
\label{fig:mc_n}
\end{figure}

The results are summarized in Fig.~\ref{fig:mc}, where the complexity 
of each policy is also labeled. In order to compare the rate-function 
$I(b)$ as defined in Eq.~(\ref{eq:rf}), we plot the probability over 
the number of channels or users, i.e., $n$, for a fixed value of 
threshold $b$. The negative of the slopes of the curves can be viewed 
as the rate-function for each policy. In Fig.~\ref{fig:mc}, we report 
the results only for $b=4$, and the results are similar for other values 
of threshold $b$. From Fig.~\ref{fig:mc}, we observe that D-SSG has a 
similar delay performance to that of DWM and Hybrid, which are both known 
to be rate-function delay-optimal. This not only supports our theoretical 
results that D-SSG guarantees a near-optimal rate-function, but also 
implies that D-SSG empirically performs very well while enjoying a lower 
complexity. Further, we observe that D-SSG consistently outperforms its 
queue-length-based counterpart, Q-SSG, despite the fact that in \cite{bodas10}, 
it has been shown through simulations that Q-SSG \emph{empirically} achieves 
near-optimal queue-length performance. This provides a further evidence 
that good queue-length performance does not necessarily translate into 
good delay performance. The results also show that D-MWS yields a zero 
rate-function, as expected.

We also plot the probability for delay threshold $b$ as in \cite{bodas09,
bodas10,bodas11a,sharma11,sharma11b,ji13a} to investigate the performance 
of different policies for fixed $n$. In Fig.~\ref{fig:mc_n}, we report the 
results for $n=10$, and the results are similar for other values of $n$. 
From Fig.~\ref{fig:mc_n}, we observe that D-SSG consistently performs closely 
to DWM and Hybrid for almost all values of $b$ that we consider. We also observe 
that D-SSG consistently outperforms its queue-length-based counterpart, Q-SSG.

In addition, we compute the average time required for the operations of 
each policy within one scheduling cycle, when $n=100$. Running simulations 
in a PC with Intel Core i7-2600 3.4GHz CPU and 8GB memory, D-SSG requires 
roughly 0.3 millisecond to finish all of the required operations within 
one scheduling cycle (which, for example, is 1 millisecond in LTE systems), 
while the two-stage Hybrid policy needs 7-10 times more. This, along with 
the above simulation results, implies that in practice D-SSG is more suitable 
for actual implementations than the hybrid policies, although D-SSG does not 
guarantee rate-function delay optimality.

\begin{figure}[t]
\centering
\epsfig{file=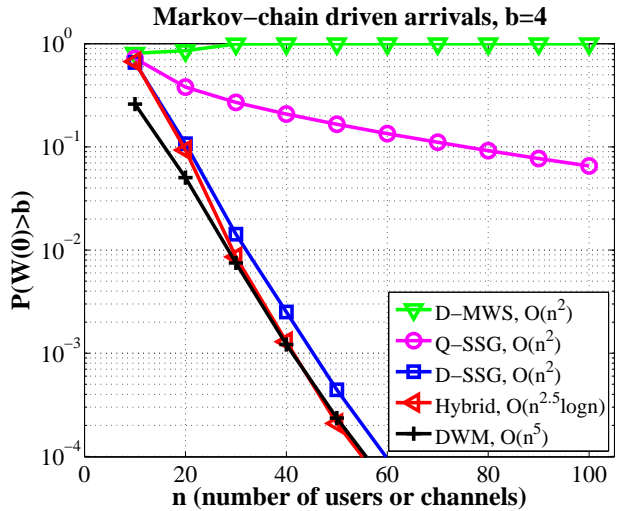,width=0.5\linewidth}
\caption{Performance comparison of different scheduling policies in the 
case with heterogeneous users and Markov-chain driven channels, for $b=4$.}
\label{fig:heter_greedy}
\end{figure}

Further, we evaluate scheduling performance of different policies 
in more realistic scenarios, where users are \emph{heterogeneous} 
and channels are \emph{correlated over time}. Specifically, we consider 
channels that can be modeled as a two-state Markov chain, where the 
channel is ``ON" when the Markov chain is in state 1, and is ``OFF" 
when it is in state 2. We assume that there are two classes of 
users: users with an odd index are called \emph{near-users}, and 
users with an even index are called \emph{far-users}. Different
classes of users see different channel conditions: near-users see 
better channel condition, and far-users see worse channel condition. 
We assume that the transition probability matrices of channels for 
near-users and far-users are $[0.833, 0.167; 0.5 , 0.5]$ and 
$[0.5, 0.5; 0.167, 0.833]$, respectively. The arrival processes are 
assumed to be the same as in the previous case. 

The results are summarized in Fig.~\ref{fig:heter_greedy}. We observe 
similar results as in the previous case with homogeneous users and 
\emph{i.i.d.} channels in time. In particular, D-SSG exhibits a 
rate-function that is similar to that of DWM and Hybrid, although 
its delay performance is slightly worse. Note that in this scenario, 
a rate-function delay-optimal policy is \emph{not} known yet. Hence, 
for future work, it would be interesting to understand how to design 
rate-function delay-optimal or near-optimal policies in general scenarios.

\section{Conclusion} \label{sec:con}
In this paper, we developed a practical and low-complexity greedy scheduling 
policy (D-SSG) that not only achieves throughput optimality, but also guarantees 
a near-optimal delay rate-function, for multi-channel wireless networks. Our 
studies reveal that throughput optimality is relatively easier to achieve in 
such multi-channel systems, while there exists an explicit trade-off between 
complexity and delay performance. If one can bear a minimal drop in the delay 
performance, lower-complexity scheduling policies can be exploited.

The analytical results in this paper are derived for the \emph{i.i.d.} ON-OFF 
channel model with unit channel capacity. An interesting direction for future 
work is to study general multi-rate channels that can be correlated over time. 
We note that this problem will become much more challenging. For example, even 
for an \emph{i.i.d.} 0-$K$ channel model with channel capacity $K>1$, it is 
still unclear whether there exists a scheduling policy that can guarantee both 
optimal throughput and optimal/near-optimal asymptotic delay performance. Another 
direction for future work is to consider heterogeneous users with different 
arrival processes and different delay requirements. In these more general 
scenarios, it may be worth exploring how to find efficient schedulers that can 
guarantee a nontrivial lower bound of the optimal rate-function, if it turns out 
to be too difficult to achieve or prove the optimal asymptotic delay performance 
itself. Nonetheless, we believe that the results derived in this paper will 
provide useful insights for designing high-performance scheduling policies for 
more general scenarios.



\appendices 
\section{Proof of Theorem~\ref{thm:ub}}  \label{app:thm:ub}
We consider event $\mE_1$ and a sequence of events $\mE^c_2$ implying 
the occurrence of event $\{W(0)>b\}$. 

{\bf Event} $\mE_1$: Suppose that there is a packet that arrives to 
the network in time-slot $-b-1$. Without loss of generality, we assume
that the packet arrives to queue $Q_1$. Further, suppose that $Q_1$ is 
disconnected from all the $n$ servers in all the time-slots from $-b-1$ 
to $-1$. 

Then, at the beginning of time-slot 0, this packet is still in the 
network and has a delay of $b+1$. This implies $\mE_1 \subseteq \{W(0)>b\}$. 
Note that the probability that event $\mE_1$ occurs can be computed as
\[
\Prob(\mE_1) = (1-q)^{n(b+1)}=e^{-n(b+1)I_X}.
\]
Hence, we have 
\[
\Prob(W(0)>b) \ge e^{-n(b+1)I_X},
\]
and thus
\[
\limsup_{n \rightarrow \infty} \frac {-1}{n} \log \Prob(W(0)>b) \le (b+1)I_X.
\]

{\bf Event} $\mE^c_2$: Consider any fixed $c \in \{0,1,\dots,b\}$. 
Fix any $\epsilon>0$, and choose $t$ such that $I_{AG}(t,b-c) \le
I_{AG}(b-c)+\epsilon$. Suppose that from time-slot $-t-b$ to $-b-1$, 
the total number of packet arrivals to the system is greater than 
$nt + n(b-c)$, and let $p_{(b-c)}$ denote the probability that this
event occurs. Then, from the definitions of $I_{AG}(t,x)$ and 
$I_{AG}(t)$, we know
\[
\liminf_{n \rightarrow \infty} \frac {-1}{n} \log p_{(b-c)} 
= I_{AG}(t,b-c) \le I_{AG}(b-c)+\epsilon. 
\]
Clearly, the total number of packets that are served in any time-slot 
is no greater than $n$. Hence, at the end of time-slot $-b-1$, there 
are at least $n(b-c)+1$ packets remaining in the system. Moreover, at 
the end of time-slot $-c-1$, the system contains at least one packet 
that arrived before time-slot $-b$. Without loss of generality, we assume 
that this packet is in $Q_1$. Now, assume that $Q_1$ is disconnected from 
all the $n$ servers in the next $c$ time-slots, i.e., from time-slot $-c$ 
to $-1$. This occurs with probability $(1-q)^{cn}=e^{-ncI_X}$, independently 
of all the past history. Hence, at the beginning of time-slot 0, there is 
still a packet that arrived before time-slot $-b$. Hence, we have $W(0)>b$ 
in this case. This implies $\mE^c_2 \subseteq \{W(0)>b\}$. Note that the 
probability that event $\mE^c_2$ occurs can be computed as
\[
\Prob(\mE^c_2) = p_{(b-c)}e^{-ncI_X}.
\]
Hence, we have 
\[
\Prob(W(0)>b) \ge p_{(b-c)}e^{-ncI_X},
\]
and thus
\[
\limsup_{n \rightarrow \infty} \frac {-1}{n} \log \Prob(W(0)>b) \le I_{AG}(b-c) + \epsilon + cI_X.
\]
Since the above inequality holds for any $c \in \{0,1,\dots,b\}$ and all $\epsilon>0$, 
by letting $\epsilon$ tend to 0 and taking the minimum over all $c \in \{0,1,\dots,b\}$, 
we have
\[
\begin{split}
& \limsup_{n \rightarrow \infty} \frac {-1}{n} \log \Prob(W(0)>b) \\
& \le \min_{c \in \{0,1,\dots,b \}} \{I_{AG}(b-c) + cI_X \}.
\end{split}
\]

Considering both event $\mE_1$ and events $\mE^c_2$, we have
\[
\begin{split}
& \limsup_{n \rightarrow \infty} \frac {-1}{n} \log \Prob(W(0)>b) \\
& \le \min \{\min_{c \in \{0,1,\dots,b \}} \{I_{AG}(b-c) + cI_X \}, (b+1)I_X\}.
\end{split}
\]

\section{Strictly Positive Rate-function}  \label{app:posi_rf}
Consider \emph{i.i.d.} 0-$L$ arrivals: in each time-slot, there are $L$ packet 
arrivals with probability $\alpha$, and no arrivals otherwise. Clearly, the mean 
arrival rate $\alpha L$ must be smaller than 1, otherwise the system is not 
stabilizable and thus the rate-function of the delay-violation probability becomes 
zero under any scheduling policy. Hence, we want to show that our derived upper 
bound $I_U(b)$ is strictly positive in the case of \emph{i.i.d.} 0-$L$ arrivals 
with mean arrival rate $\alpha L < 1$.

Recall that $I_U(b) = \min \{(b+1) I_X, \min_{0 \le c \le b} \{ I_{AG}(b-c) 
+ c I_X \} \}$, where $I_X = \log \frac {1} {1-q}$, $I_{AG}(x) = \inf_{t>0} 
I_{AG}(t,x)$, and $I_{AG}(t,x) = \liminf_{n \rightarrow \infty} \frac {-1}{n} 
\log \Prob(A(-t+1,0) > n(t+x))$. It is easy to see that $I_X > 0$ and $I_{AG}(x) 
\ge 0$ for all $x \ge 0$. Hence, in order to show $I_U(b)>0$ it suffices to 
show that $I_{AG}(b) > 0$. Also, note that $I_{AG}(t,x)$ is a non-decreasing 
function of $x$, and thus $I_{AG}(x)$ is also non-decreasing. Therefore, it 
suffices to show $I_{AG}(0) > 0$. Let $N(\tau)$ denote the number of queues
that have arrivals in time-slot $\tau$, and let $N(t_1,t_2) = \sum_{\tau=t_1}^{t_2}
N(\tau)$. Also, let $\alpha^{\prime} = \frac {1}{L}$. Then, we have
\[
\begin{split}
&\Prob(A(-t+1,0) > nt) \\
\le & \Prob(A(-t+1,0) \ge nt) \\
= & \Prob(N(-t+1,0) \ge \alpha^{\prime} nt) \\
\le & e^{-nt D(\alpha \| \alpha^{\prime})}, 
\end{split}
\]
where the last inequality is from the Chernoff bound, and $D(x\|y) = x \log \frac {x} {y} 
+ (1-x) \log \frac {1-x} {1-y}$, is the Kullback-Leibler divergence. 
Hence, $I_{AG}(t,0) = \liminf_{n \rightarrow \infty} \frac {-1}{n} \log \Prob(A(-t+1,0) > nt) 
\ge t D(\alpha \| \alpha^{\prime})$, and $I_{AG}(0) = \inf_{t>0} I_{AG}(t,0) \ge D(\alpha \| 
\alpha^{\prime}) > 0$. This completes the proof.


\section{Delay-based Queue-Side-Greedy (D-QSG) And Sample-Path Equivalence} \label{app:sec:dqsg}
Delay-based Queue-Side-Greedy (D-QSG) policy, in an iterative manner too, schedules the oldest 
packets in the system one-by-one whenever possible. In this sense, D-QSG can be viewed 
as an intuitive approximation of the Oldest Packet First (OPF) policies\footnote{A 
scheduling policy $\mathbf{P}$ is said to be an OPF policy if in any time-slot, policy 
$\mathbf{P}$ can serve the $k$ oldest packets in that time-slot for the largest possible 
value of $k \in \{1,2,\dots,n\}$.} \cite{ji13a} that have been shown to be rate-function 
delay-optimal. We later prove an important sample-path equivalence result 
(Lemma~\ref{lem:equivalent}) that will be used in proving Lemma~\ref{lem:success}
and thus the main result of this paper (Theorem~\ref{thm:dssg-lb}). 

We start by presenting some additional notations. In the D-QSG policy, 
there are at most $n$ rounds in each time-slot $t$. By slightly abusing
the notations, we let $Q_i^k(t)$, $Z^k_{i,l}(t)$ and $W_i^k(t)=Z^k_{i,1}(t)$ 
denote the length of queue $Q_i$, the delay of the $l$-th packet of $Q_i$, 
and the HOL delay of $Q_i$ after the $k$-th round in time-slot $t$ under 
D-QSG, respectively. Let $\Upsilon_k(t)$ denote the set of indices of the 
available servers at the beginning of the $k$-th round, and let $\Psi_k(t)$ 
denote the set of queues that have the largest HOL delay among all the 
queues that are connected to at least one server in $\Upsilon_k(t)$ at 
the beginning of the $k$-th round, i.e., $\Psi_k(t) \triangleq \{ 1 \le 
i \le n ~|~ W_i^{k-1}(t) \cdot \mathbb{1}_{\{\sum_{j \in \Upsilon_k(t)} 
C_{i,j}(t) > 0\}} = \max_{1 \le l \le n} W_l^{k-1}(t) \cdot \mathbb{1}_{\{\sum_{j 
\in \Upsilon_k(t)} C_{l,j}(t) > 0\}} \}$. Also, let $i(k,t)$ be the index 
of the queue that is served in the $k$-th round of time-slot $t$, and let 
$j(k,t)$ be the index of the server that serves $Q_{i(k,t)}$ in that round.
We then specify the operations of D-QSG as follows.

\noindent {\bf Delay-based Queue-Side-Greedy (D-QSG) policy:}
In each time-slot $t$, 

\begin{enumerate}
\item 
Initialize $k=1$ and $\Upsilon_1 = \{1,2,\dots,n\}$.
\item In the $k$-th round, allocate server $S_{j(k,t)}$ to $Q_{i(k,t)}$, where 
\[
\begin{split}
i(k,t) &= \min \{i ~|~ i \in \Psi_k(t) \}, \\
j(k,t) &= \min {\{j \in \Upsilon_k(t) ~|~ C_{i(k,t),j}(t)=1\}}.
\end{split}
\]
That is, in the $k$-th round, we consider the queues that have the largest HOL 
delay among those that have at least one available server connected (i.e., the 
queues in set $\Psi_k(t)$), and break ties by picking the queue with the smallest 
index (i.e., $Q_{i(k,t)}$). We then choose an available server that is connected 
to queue $Q_{i(k,t)}$, and break ties by picking the server with the smallest index
(i.e., server $S_{j(k,t)}$), to serve $Q_{i(k,t)}$. At the end of the $k$-th round, 
update the length of $Q_{i(k,t)}$ to account for service, i.e., set $Q_{i(k,t)}^k(t) 
= \left(Q_{i(k,t)}^{k-1}(t) - C_{i(k,t),j(k,t)}(t) \right)^+$ and $Q_i^k(t) = 
Q_i^{k-1}(t)$ for all $i \neq i(k,t)$. Also, update the HOL delay of $Q_{i(k,t)}$, 
by setting $W_{i(k,t)}^k(t) = Z^k_{i(k,t),1}(t) =  Z^{k-1}_{i(k,t),2}(t)$ if 
$Q_{i(k,t)}^k(t)>0$, and $W_{i(k,t)}^k(t)=0$ otherwise, and setting $W_i^k(t) 
= W_i^{k-1}(t)$ for all $i \neq i(k,t)$. 

\item Stop if $k$ equals $n$. Otherwise, increase $k$ by 1, set $\Upsilon_k(t) = 
\Upsilon_{k-1}(t) \backslash \{j(k,t)\}$, and repeat step 2.
\end{enumerate} 
 
\emph{Remark:} 
Note that D-QSG is only used for assisting the rate-function delay 
analysis of D-SSG and may not be suitable for practical implementation 
due to its $O(n^3)$ complexity. This is because there are at most $n$ 
rounds, and in each round, it takes $O(n^2 + n) = O(n^2)$ time to find 
a queue that has at least one connected and available server (which 
takes $O(n^2)$ time to check for all queues) and that has the largest 
HOL delay (which takes $O(n)$ time to compare).


The following lemma  states the sample-path equivalence property between 
D-QSG and D-SSG under the tie-breaking rules specified in this paper.

\begin{lemma}
\label{lem:equivalent}
For the same sample path, i.e., same realizations of arrivals and channel 
connectivity, D-QSG and D-SSG pick the same schedule in every time-slot.
\end{lemma}
\begin{proof}
We prove Lemma~\ref{lem:equivalent} by induction.
It suffices to prove that for any given system, i.e., for any given set of 
packets after arrivals and for any channel realizations, both D-SSG and D-QSG 
pick the same schedule. Suppose that there are $K$ packets in the system. Let 
$x_k$ denote the $k$-th oldest packet in the system. We want to show that packet 
$x_k$ is either served by the same server under both D-SSG and D-QSG, or is not 
served by any server under both D-SSG and D-QSG. We denote the set of the $k$ 
oldest packets by $\mP_k \triangleq \{x_r ~|~ r \le k\}$, and denote the set of 
the first $k$ servers by $\mS_k \triangleq \{S_j ~|~ j \le k \}$. Let $S_{j(r)}$ 
denote the server allocated to serve the $r$-th oldest packet under D-QSG. We 
prove it by induction method.

{\bf Base case:}
Consider packet $x_1$, i.e., the oldest packet, and consider two cases:
under D-QSG, 1) packet $x_1$ is served by $S_{j(1)}$; 2) packet $x_1$ is 
not served by any server.

In Case 1), we want to show that packet $x_1$ is also served by the same 
server $S_{j(1)}$ under D-SSG. Note that packet $x_1$ is the oldest packet 
in the system and is the first packet to be considered under D-QSG. Since 
it is served by $S_{j(1)}$, from the tie-breaking rule of D-QSG, we know  
that the queue that contains packet $x_1$ is disconnected from all the 
servers in set $\mS_{j(1)}$ except server $S_{j(1)}$. Now, we consider the 
server allocation under D-SSG, which allocates servers one-by-one in an 
increasing order of the server index. Since all the servers in set $\mS_{j(1)}$ 
except for server $S_{j(1)}$ are disconnected from the queue containing packet 
$x_1$, these servers cannot be allocated to packet $x_1$ in the first $j(1)-1$ 
rounds under D-SSG. While in the $j(1)$-th round, D-SSG must allocate server 
$S_{j(1)}$ to packet $x_1$, since the queue that contains packet $x_1$ is the 
queue that has the largest HOL among the queues that are connected to server 
$S_{j(1)}$.

In Case 2), packet $x_1$ is the first packet to be considered under D-QSG,
but is not served by any server. This implies that no servers are connected 
to the queue that contains packet $x_1$. Hence, packet $x_1$ cannot be served 
under D-SSG either.

Combining the above two cases, we prove the base case.

{\bf Induction step:}
Consider an integer $k \in \{1,2,\dots, K-1\}$. Suppose that every
packet in set $\mP_k$ is either served by the same server under 
both D-QSG and D-SSG, or is not served by any server under both 
D-QSG and D-SSG. We want to show that this also holds for every
packet in set $\mP_{k+1}$. Clearly, it suffices to consider only 
packet $x_{k+1}$ (i.e., the $(k+1)$-th oldest packet in the system), 
as the other packets all satisfy the condition from the induction 
hypothesis. We next consider two cases: under D-QSG, 1) packet 
$x_{k+1}$ is scheduled by a server under D-QSG; 2) $x_{k+1}$ is 
not served by any server.

In Case 1), suppose that packet $x_{k+1}$ is served by server $S_{j(k+1)}$ 
under D-QSG. We want to show that packet $x_{k+1}$ is also served by 
server $S_{j(k+1)}$ under D-SSG. We first show that under D-SSG, packet 
$x_{k+1}$ cannot be served in the first $j(k+1)-1$ rounds. Note that 
under D-QSG, packet $x_{k+1}$ is served by server $S_{j(k+1)}$. This 
implies that any server in set $\mS_{j(k+1)-1}$ is either disconnected 
from the queue that contains packet $x_{k+1}$ or has already been allocated 
to packets in set $\mP_k$ under D-QSG. This, along with the induction 
hypothesis, further implies that under D-SSG, in the first $j(k+1)-1$ 
rounds, the servers under consideration are either disconnected from 
the queue that contains packet $x_{k+1}$ or allocated to packets in 
set $\mP_k$. Hence, packet $x_{k+1}$ cannot be scheduled in the 
first $j(k+1)-1$ rounds under D-SSG. Next, we want to show that 
packet $x_{k+1}$ must be served by server $S_{j(k+1)}$ in the 
$j(k+1)$-th round under D-SSG. Let $\mP^{\prime}_k \subseteq \mP_k$ 
denote the set of packets among the $k$ oldest packets that are not 
served under both D-QSG and D-SSG. Then, all the queues that contain
packets in set $\mP^{\prime}_k$ must be disconnected from server $S_{j(k+1)}$, 
otherwise some packet $x_r \in \mP^{\prime}_k$ should be served by 
server $S_{j(k+1)}$ under D-QSG. On the other hand, the induction 
hypothesis implies that any packet $x_r \in \mP_k \backslash \mP^{\prime}_k$ 
must be served by some server $S_{j(r)}$, under D-SSG, 
where $j(r) \neq j(k+1)$. Hence, D-SSG does not allocate server 
$S_{j(k+1)}$ to any packet in set $\mP_k$. Therefore, in the $(k+1)$-th 
round, D-SSG must allocate server $S_{j(k+1)}$ to packet $x_{k+1}$, 
since the queue that contains packet $x_{k+1}$ has the largest HOL 
delay among the queues that are connected to server $S_{j(k+1)}$.

In Case 2), packet $x_{k+1}$ is not served by any server under D-QSG. 
This implies that the queue that contains packet $x_{k+1}$ is disconnected 
from all the servers in set $\mS_n \backslash \{S_{j(r)} ~|~ r \in \mP_k 
\backslash \mP^{\prime}_k \}$, i.e., the set of available servers when 
considering packet $x_{k+1}$. On the other hand, the induction hypothesis 
implies that under D-SSG, all the servers in set $\{S_{j(r)} ~|~ r \in 
\mP_k \backslash \mP^{\prime}_k \}$ are also allocated to packets in set 
$\mP_k \backslash \mP^{\prime}_k$. Hence, packet $x_{k+1}$ cannot be served 
by any server under D-SSG either.

Combining the above two cases, we prove the induction step. This
completes the proof.
\end{proof}

Note that under D-SSG, in each round, when a server has multiple connected 
queues that have the largest HOL delay, we break ties by picking the queue 
with the smallest index. Presumably, one can take other arbitrary tie-breaking 
rules. However, it turns out to be much more difficult to directly analyze 
the rate-function performance for a greedy policy from the server side (like 
D-SSG) without using the above equivalence property. For example, as we 
mentioned earlier, the authors of \cite{bodas10,bodas11a} were only able to 
prove a positive (queue-length) rate-function for Q-SSG in more restricted 
scenarios. Hence, our choice of the above simple tie-breaking rule is in fact 
quite important for proving the above sample-path equivalence result, which 
in turn plays a critical role in proving a key property of D-SSG (Lemma~\ref{lem:success}) 
and thus near-optimal rate-function of D-SSG (Theorem~\ref{thm:dssg-lb}). 
Nevertheless, we would expect that one can choose arbitrary tie-breaking 
rules for D-SSG in practice.

\section{Proof of Lemma~\ref{lem:success}}  \label{app:lem:success}
We then divide the proof into two parts (Lemmas~\ref{lem:first} and \ref{lem:atleast}).

\begin{lemma}
\label{lem:first}
Consider a set of $n$ packets. Consider any function $f(n) < \frac {n}{2}$, 
which is strictly increasing with $n$. The D-SSG policy is applied to schedule 
these $n$ packets. Then, there exists a finite integer $N_{X1}>0$ such that 
for all $n \ge N_{X1}$, with probability no smaller than $1-(1-q)^{n-f(n)\log^2n}$, 
D-SSG schedules all the oldest $f(n)$ packets among the $n$ packets.
\end{lemma}

\begin{proof}
Since Lemma~\ref{lem:first} is focused on the oldest $f(n)$ packets in the set, 
it is easier to consider the D-QSG policy instead, which in an iterative manner 
schedules the oldest packets first. Due to the sample-path equivalence between 
D-SSG and D-QSG (Lemma~\ref{lem:equivalent}), it is sufficient to prove that the 
result of Lemma~\ref{lem:first} holds for D-QSG.

Suppose that the oldest $f(n)$ packets among the $n$ packets are from $k$ different 
queues, where $k \le f(n)$. It is easy to see that if each of the $k$ queues is 
connected to no less than $f(n)$ servers, then all of these oldest $f(n)$ packets 
will be served. Specifically, because D-QSG gives a higher priority to an older packet, 
the above condition guarantees that when D-QSG schedules any of the oldest $f(n)$ 
packets, there will always be at least one available server that is connected to 
the queue containing this packet.

Now, consider any queue $Q_i$. We want to compute the probability that
$Q_i$ is connected to no less than $f(n)$ servers. We first compute the
probability that $Q_i$ is connected to less than $f(n)$ servers:
\[
\begin{split}
\Prob( & Q_i~\text{is connected to less than}~f(n)~\text{servers})\\
= & \sum_{j=0}^{f(n)-1} \Prob( Q_i~\text{is connected to}~j~\text{servers})\\
= & \sum_{j=0}^{f(n)-1} \binom{n} {j} q^{j} (1-q)^{n-j} \\
\le & f(n) n^{f(n)} (1-q)^{n-f(n)}.
\end{split}
\] 

Next, choose $N_{X1}$ such that $f(n) n^{f(n)} < (\frac {1}{1-q})^{n-f(n)}$ 
and $f^2(n)n^{f(n)} \le (\frac {1}{1-q})^{f(n)(\log^2n-1)}$ for all 
$n \ge N_{X1}$. Such an $N_{X1}$ exists because 
$\log (f(n) n^{f(n)}) = \Theta(f(n) \log n)$ and $\log ((\frac {1}{1-q})^{n-f(n)}) 
= \Theta(n)$, hence, we have $\log (f(n) n^{f(n)}) < \log((\frac {1}{1-q})^{n-f(n)})$
and thus $f(n) n^{f(n)} < (\frac {1}{1-q})^{n-f(n)}$ for large enough $n$;
and similarly, because $\log (f^2(n)n^{f(n)}) = \Theta(f(n) \log n)$ and $\log 
((\frac {1}{1-q})^{f(n)(\log^2n-1)}) = \Theta(f(n) \log^2n)$, hence, we have
$\log (f^2(n)n^{f(n)}) \le \log ((\frac {1}{1-q})^{f(n)(\log^2n-1)})$ and 
thus $f^2(n)n^{f(n)} \le (\frac {1}{1-q})^{f(n)(\log^2n-1)}$ for large enough
$n$. Then, the probability that each of the $k$ queues is connected to no less 
than $f(n)$ servers is:
\[
\begin{split}
\Prob( & \text{Each of the}~k~\text{queues is connected to no less than}~f(n)~\text{servers})\\
\ge & (1 - f(n) n^{f(n)} (1-q)^{n-f(n)})^k \\
\stackrel{(a)} \ge & (1 - f(n) n^{f(n)} (1-q)^{n-f(n)})^{f(n)} \\
\stackrel{(b)} \ge & 1 - f^2(n) n^{f(n)} (1-q)^{n-f(n)} \\
\stackrel{(c)} \ge & 1 - (1-q)^{n-f(n)\log^2 n}
\end{split}
\] 
for all $n \ge N_{X1}$, where (a) is from our choice of $N_{X1}$ and the 
fact that $k \le f(n)$, (b) is from our choice of $N_{X1}$ and Bernoulli's 
inequality (i.e., $(1+x)^r \ge 1+rx$ for every real number $x \ge -1$ and 
every integer $r \ge 0$), and (c) is from our choice of $N_{X1}$. This 
completes the proof.
\end{proof}

\begin{lemma}
\label{lem:atleast}
Consider a set of $n$ packets satisfying that no more than $2H$ packets 
are from the same queue, where $H>4$ is any integer constant independent
of $n$. The D-SSG policy is applied to schedule these $n$ packets. Then, 
there exists a finite integer $N_{X2}>0$ such that for all $n \ge N_{X2}$, 
with probability no smaller than $1-(1-q)^n$, D-SSG schedules at least 
$n-H\sqrt{n}$ packets among the $n$ packets.
\end{lemma}

\begin{proof} 
Consider the D-SSG policy. We first compute the probability that some $H\sqrt{n}$ 
packets are not scheduled by D-SSG, which is equivalent to the event that some 
$H\sqrt{n}$ servers are not allocated to any packet by D-SSG. 

Consider any arbitrary set of servers $\Xi = \{S_{r_j} ~|~ j=1,2,\dots,H\sqrt{n}\}$,
where $r_i < r_j$ if $i < j$.
Clearly, we have $r_j \le n-H\sqrt{n}+j$ for all $j \in \{1,2,\dots,H\sqrt{n}\}$.
Consider the $r_j$-th server $S_{r_j}$. Then, the number of remaining packets is at least 
$n-r_j+1$ at the beginning of the $r_j$-th round. Since no more than $2H$ packets are from 
the same queue, there are at least $\lceil \frac {n-r_j+1} {2H} \rceil$ queues that are 
non-empty at the beginning of the $r_j$-th round. Then, the probability that server $S_{r_j}$ 
is not allocated to any packet is no greater than $(1-q)^{\lceil \frac{n-r_j+1} {2H} \rceil} 
\le (1-q)^{\frac{n-r_j+1} {2H}}$. Hence, 
\[
\begin{split}
\Prob(& \text{None of the servers in a given set}~\Xi~\text{is allocated}) \\
\le & \prod_{j=1}^{H\sqrt{n}} (1-q)^{\frac {n-r_j+1} {2H}} \\
\le & \prod_{j=1}^{H\sqrt{n}} (1-q)^{\frac {n-(n-H\sqrt{n}+j)+1} {2H}} \\
\le & (1-q)^{\frac {1} {2H} (1+2+\dots+H\sqrt{n})} \\
= & (1-q)^{\frac {H}{4} n + \frac {\sqrt{n}} {4}} \\
\end{split}
\]

Since $H>4$, there exists an $N_{X2}$ such that $n^{H\sqrt{n}} (1-q)^{\frac {H}{4} n 
+ \frac {\sqrt{n}} {4}} \le (1-q)^n$ for all $n \ge N_{X2}$. Such an $N_{X2}$ exists 
because $\log (n^{H\sqrt{n}}) = \Theta (\sqrt{n} \log n)$ and $\log ((\frac 
{1}{1-q})^{\frac {H}{4} n + \frac {\sqrt{n}} {4} - n}) = \Theta (n^{\frac{H}{4}-1})$,
hence, $\log (n^{H\sqrt{n}}) \le \log ((\frac {1}{1-q})^{\frac {H}{4} n + \frac {\sqrt{n}} 
{4} - n})$ and thus $n^{H\sqrt{n}} (1-q)^{\frac {H}{4} n + \frac {\sqrt{n}} {4}} \le (1-q)^n$ 
for large enough $n$. Then, we can compute the probability that some $H\sqrt{n}$ servers 
are not allocated as
\[
\begin{split}
\Prob(& \text{Some}~H\sqrt{n}~\text{servers are not allocated}) \\
\le & \binom{n} {H\sqrt{n}}\Prob(\text{None of the servers in a given set}~\Xi~\text{is allocated}) \\
\le & n^{H\sqrt{n}} (1-q)^{\frac {H}{4} n + \frac {\sqrt{n}} {4}} \\
\le & (1-q)^n
\end{split}
\]
for all $n \ge N_{X2}$, where the last inequality is due to our choice of $N_{X2}$. 

Therefore, we have
\[
\begin{split}
\Prob(& \text{At least}~n-H\sqrt{n}~\text{packets are scheduled}) \\
= & 1 - \Prob(\text{Less than}~n-H\sqrt{n}~\text{packets are scheduled}) \\
= & 1 - \Prob(\text{Greater than}~H\sqrt{n}~\text{packets are not scheduled}) \\
\ge & 1 - \Prob(\text{At least}~H\sqrt{n}~\text{packets are not scheduled}) \\
= & 1 - \Prob(\text{Some}~H\sqrt{n}~\text{packets are not scheduled}) \\
= & 1 - \Prob(\text{Some}~H\sqrt{n}~\text{servers are not allocated}) \\
\ge & 1-(1-q)^n,
\end{split}
\]
for all $n \ge N_{X2}$.
\end{proof}

By applying Lemmas~\ref{lem:first} and \ref{lem:atleast}, and choosing 
$N_X \triangleq \max \{N_{X1}, N_{X2}, N_{X3}\}$, where $N_{X3}$ is such
that $n-H\sqrt{n} > f(n)$ for all $n \ge N_{X3}$, we show that for all 
$n \ge N_X$, with probability no smaller than $1-2(1-q)^{n-f(n)\log^2 n}$, 
D-SSG schedules at least $n-H\sqrt{n}$ packets including the oldest $f(n)$ 
packets among the $n$ packets.

\section{Proof of Theorem~\ref{thm:g-fbs}}  \label{app:thm:g-fbs}
The proof follows a similar argument for the proof of Theorem~2 in 
\cite{sharma11b}. 

We start by defining $I_0 \triangleq I_U(b-1)=\min \{b I_X, \min_{0 \le c 
\le b-1} \{ I_{AG}(b-1-c) + c I_X \} \} \ge I^*(b-1)$. Consider any fixed 
$\epsilon>0$, and define $I^\epsilon_0 \triangleq \min \{b I_X, \min_{0 
\le c \le b-1} \{ I_{AG}(b-1-c) - \epsilon + c I_X \} \}$. Then, we have 
$\lim_{\epsilon \rightarrow 0} I^\epsilon_0 = I_0$. 

We then choose the value of parameter $h$ for G-FBS based on the statistics 
of the arrival process. We fix $\delta<\frac{2}{3}$ and $\eta<\frac{p}{2}$. 
Then, from Assumption~\ref{ass:arr_ld}, there exists a positive function 
$I_B(\eta,\delta)$ such that for all $n \ge N_B(\eta,\delta)$ and $t \ge 
T_B(\eta,\delta)$, we have
\[
\Prob \left( \frac {\sum_{\tau=r+1}^{r+t} \mathbb{1}_{\{|A(\tau) - pn|
> \eta n \}}} {t} > \delta \right) < \exp (-nt I_B(\eta,\delta)),
\]
for any integer $r$. We then choose 
\[
h = \max \left \{ T_B(\eta,\delta), 
\left \lceil \frac {1}{(p-\eta)(1-\frac{3\delta}{2})} \right \rceil,  
\left \lceil \frac {2I^\epsilon_0}{I_B(\eta,\delta)} \right \rceil, 4 \right \} + 1.
\]
The reason for choosing the above value of $h$ will become clear later on.
Recall from Assumption~\ref{ass:arr_bound} that $L$ is the maximum number
of packets that can arrive to a queue in any time-slot $t$. Then, $H=Lh$ is 
the maximum number of packets that can arrive to a queue during an interval 
of $h$ time-slots, and is thus the maximum number of packets from the same 
queue in a frame. This also implies that in the S-frame, no more than $2H$ 
packets are from the same queue.

Next, we define the following notions associated with the G-FBS policy. 
Let $F(t)$ denote the number of unserved frames in time-slot $t$, and let 
$R(t)$ denote the remaining available space (where the unit is packet) in 
the end-of-line frame at the end of time-slot $t$. Also, let $X_F(t)$ denote 
the indicator function of whether a success occurs in time-slot $t$. That 
is, $X_F(t)=1$ if there is a success, and $X_F(t)=0$ otherwise. Recall that 
$n_0=n-H\sqrt{n}$. Then, we can write a recursive equation for $F(t)$:
\begin{align}
F(t) &= (F(t-1) + \left \lceil \frac {A(t)-R(t-1)} {n_0} \right \rceil 
- X_F(t),0)^+, \\
R(t) &= \mathbb{1}_{ \{ F(t)>0 \} } \cdot ((R(t-1)-A(t)) \mod n_0 ).
\end{align}

Let $M(t) \le H\sqrt{n}$ denote the number of packets in the L-frame at 
the beginning of time-slot $t$, and let $P(t) \le n_0$ denote the number 
of packets in the HOL frame at the beginning of time-slot $t$. Then, at 
the beginning of time-slot $t$, the number of packets in the S-frame is 
equal to $M(t)+P(t)$. Let $D(t) \le M(t) + P(t)$ denote the number 
of packets served from the S-frame if a success occurs in time-slot $t$. 
Then, we have the following recursive equation for $M(t)$:
\[
M(t+1) = \left\{
\begin{array}{ll}
M(t) + P(t) - D(t), & ~\text{if}~ X_F(t)=1, \\
M(t), & ~\text{otherwise}.
\end{array}
\right.
\]
Also, we let 
\[
X_F(t_1,t_2)=\sum^{t_2}_{\tau=t_1} X_F(\tau) \mathbb{1}_{ \{ \{F(\tau)>0 \} \cup \{ M(\tau)>0 \} \} }
\]
denote the the total number of successes in the interval from time-slot 
$t_1$ to $t_2$ when the S-frame is non-empty (i.e., the number of unserved 
frames is greater than zero or the L-frame is non-empty). 

Note that the arriving time of a frame is the time when its first packet 
arrives. Let $R_0=R(t_1-1)$ denote the available space in the end-of-line 
frame at the end of time-slot $t_1-1$. Then, we let $A^{R_0}_F(t_1,t_2)$ 
denote the number of new frames that arrive from time-slot $t_1$ to $t_2$. 
When $R_0=0$, we use $A_F(t_1,t_2)$ to denote $A^{R_0}_F(t_1,t_2)$ for 
notational convenience.

Let $L(-b)$ be the last time before $-b$, when the number of unserved frames 
is equal to zero. Then, given that $L(-b)=-t-b-1$, where $t > 0$, the number 
of unserved frames never becomes zero during interval $[-t-b, -b-1]$. Let 
$U(0)$ denote the indicator function of whether at time-slot 0 the L-frame 
contains a packet that arrives before time-slot $-b$, i.e., $U(0)=1$ if at 
time-slot 0 the L-frame contains a packet that arrives before time-slot $-b$, 
and $U(0)=0$, otherwise. Let $\mE^{\alpha_1}_t$ denote the event that the 
number of frames that arrive during interval $[-t-b, -b-1]$ is greater than 
the total number of successes during interval $[-t-b, -1]$ when the S-frame 
is non-empty, i.e., 
\[
\mE^{\alpha_1}_t = \{A_F(-t-b,-b-1) > X_F(-t-b,-1) \}.
\]
Let $\mE^{\alpha_2}_t$ denote the event that the number of frames that arrive 
during interval $[-t-b, -b-1]$ is equal to the total number of successes during 
interval $[-t-b, -1]$ when the S-frame is non-empty, and at time-slot $0$ the 
L-frame contains a packet that arrives before time-slot $-b$, i.e., 
\[
\mE^{\alpha_2}_t = \{A_F(-t-b,-b-1) = X_F(-t-b,-1), U(0)=1\}.
\]
Letting $\mE^{\alpha}_t = \mE^{\alpha_1}_t \cup \mE^{\alpha_2}_t$, we have 
\begin{equation}
\label{eq:eq_events}
\begin{split}
\{L(-b) & =-t-b-1, W(0)>b\} \\
&= \{L(-b)=-t-b-1, \mE^{\alpha}_t\}.
\end{split}
\end{equation}
By taking the union over all possible values of $L(-b)$ and applying 
the union bound, we have
\begin{equation}
\label{eq:pwgtb}
\Prob( W(0)>b) \le \sum_{t=1}^{\infty} \Prob(L(-b)=-t-b-1, \mE^{\alpha}_t).
\end{equation}

We fix a finite time $t^*$ as 
\begin{equation}
\label{eq:tstar}
t^* \triangleq \max \{ T_1, \left \lceil \frac 
{I^\epsilon_0} { I_{BX} } \right \rceil \}, 
\end{equation}
where 
\begin{equation}
\label{eq:t1}
T_1 \triangleq \max \{
T_B(\hat{p}-p, \frac {1-\hat{p}} {6(L+2)}), \frac {18(1+\hat{p})} {(2+\hat{p})(1-\hat{p})} \}
\end{equation}
and 
\begin{equation}
\label{eq:ibx}
I_{BX} \triangleq \min \{ \frac {(1-\hat{p})I_X}{9}, I_B(\hat{p}-p,\frac {1-\hat{p}} {6(L+2)}) \}.
\end{equation}
Then, we split the summation in (\ref{eq:pwgtb}) as
\[
\Prob (W(0)>b) \le P_1 + P_2,
\]
where 
\[
P_1 \triangleq \sum_{t=1}^{t^*} \Prob(L(-b)=-t-b-1, \mE^{\alpha}_t),
\]
\[
P_2 \triangleq \sum_{t=t^*}^{\infty} \Prob(L(-b)=-t-b-1, \mE^{\alpha}_t).
\]

We divide the proof into two parts. 
In Part 1, we show that there exist a constant $C_1>0$ and a finite $N_1>0$ 
such that for all $n \ge N_1$, we have
\[
P_1 \le (C_1 + e^{g(n)}) t^* e^{-n I^\epsilon_0},
\]
where $g(n)$ is a function satisfying that $g(n) \in \omega (f(n)\log^2n)$ 
and $g(n) \in o(n)$. And in Part 2, we show that there exists a finite 
$N_2>0$ such that for all $n \ge N_2$, we have
\[
P_2 \le 4 e^{-n I^\epsilon_0}.
\]
Finally, combining both Parts, we have
\[
\Prob (W(0)>b) \le \left( (C_1 + e^{g(n)}) t^*  + 4 \right) e^{-n I^\epsilon_0},
\]
for all $n \ge N \triangleq \max \{N_1, N_2\}$. By letting $\epsilon$
tend to 0, and taking logarithm and limit as $n$ goes to infinity, we 
obtain $\liminf_{n \rightarrow \infty} \frac {-1} {n} \log \Prob 
\left( W(0) > b \right) \ge I_0$, and thus the desired results. 

We next present Lemma~\ref{lem:full} that will be used in the proof. 

\begin{lemma}
\label{lem:full}
Consider the G-FBS policy.
Let $A(\cdot)$ be an arrival process. Consider the interval $[t_1, t_2]$,
and recall that $R_0$ is the available space in the end-of-line frame at 
the end of time-slot $t_1-1$. Suppose the following two conditions are
satisfied:
\begin{enumerate}
\item The number of unserved frames never becomes zero during interval 
$[t_1,t_2]$, i.e., $F(\tau)>0$ for all $\tau \in [t_1,t_2]$.
\item For any $h-1$ consecutive time-slots in the interval, the cumulative 
arrivals are greater than or equal to $n_0$, i.e., $A(r,r+h-2) \ge n_0$, 
for any $r \in \{t_1, t_1+1, \dots, t_2-h+2\}$.
\end{enumerate}
Then the following holds,
\[
\begin{split}
A^{R_0}_F(t_1,t_2) &= \left \lceil \frac {A(t_1,t_2) - R_0} {n_0} \right \rceil 
\le \left \lceil \frac {A(t_1,t_2)} {n_0} \right \rceil, \\
R(t_1) &= (R_0 - A(t_1,t_2)) \mod n_0.
\end{split}
\]
\end{lemma}

\emph{Remarks:} The result of the above lemma implies that every frame that 
arrives in interval $[t_1,t_2]$ (except for the last frame) has exactly $n_0$ 
packets. The proof follows an inductive argument for the proof of Lemma~9 in 
\cite{sharma11b}, and is thus omitted.

{\bf Part 1:}
Consider any $t \in \{1,2,\dots,t^*\}$.
Let $\mE_t$ be the set of sample paths in which $L(-b)=-t-b-1$ and $\mE^{\alpha}_t$
occurs, let $\mE^\beta_t$ be the set of sample paths in which $\frac {A(-t-b,-b-1)} 
{n} + 1 - \sum_{\tau=-t-b}^{-1} X_F(\tau) > 0$, and let $\mE^\gamma_t$ denote the 
set of sample paths in which there are at least $n_0$ packet arrivals to the system 
during every consecutive $h-1$ time-slots in the interval of $[-t-b,-b-1]$. Let $N_0$ 
be such that $f(n)>H\sqrt{n} (t^*+b)$ for all $n \ge N_0$.

First, we want to show that for all $n \ge N_0$, we have
\begin{equation}
\label{eq:sp}
\mE_t \subseteq (\mE^\gamma_t)^c \cup \mE^\beta_t.
\end{equation}
In the following, we will show that $\mE_t \cap \mE^\gamma_t \subseteq \mE^\beta_t$, 
and then it is easy to see that $\mE_t = (\mE_t \cap \mE^\gamma_t) \cup (\mE_t \cap 
(\mE^\gamma_t)^c) \subseteq (\mE^\gamma_t)^c \cup \mE^\beta_t$.

For every sample path in set $\mE_t \cap \mE^\gamma_t$, $L(-b)=-t-b-1$ is the last 
time before $-b-1$, when the number of unserved frames is equal to zero. Then, the 
number of unserved frames never becomes empty during interval $[-t-b,-b-1]$, i.e., 
\begin{equation}
\label{eq:ftau}
F(\tau)>0~\text{for all}~\tau \in [-t-b,-b-1]. 
\end{equation}
Also, there are at least $n_0$ packet arrivals to the system during every 
consecutive $h-1$ time-slots in the interval of $[-t-b,-b-1]$. The above 
implies that the conditions of Lemma~\ref{lem:full} hold, and hence, 
\begin{equation}
\label{eq:afe}
A_F(-t-b,-b-1)=\left \lceil \frac {A(-t-b,-b-1)} {n_0} \right \rceil.
\end{equation} 
Moreover, we have $F(\tau)>0$ or $M(\tau)>0$ for all $\tau \in [-b,-1]$. Otherwise, 
there must exist one time-slot $\tau^{\prime} \in [-b,-1]$ such that $F(\tau^{\prime})=0$ 
and $M(\tau^{\prime})=0$, which implies $W(0) \le b$ and thus contradicts with 
(\ref{eq:eq_events}). This, along with (\ref{eq:ftau}), implies
\begin{equation}
\label{eq:xf}
\begin{split}
X_F(-t-b,-1) &= \sum_{\tau=-t-b}^{-1} X_F(\tau) \mathbb{1}_{\{\{F(\tau)>0\} 
\cup \{M(\tau)>0 \}\}} \\ 
&= \sum_{\tau=-t-b}^{-1} X_F(\tau).
\end{split}
\end{equation}
Hence, for any sample path in set $\mE_t \cap \mE^\gamma_t$, we have $\left 
\lceil \frac {A(-t-b,-b-1)} {n_0} \right \rceil \ge \sum_{\tau=-t-b}^{-1} 
X_F(\tau)$. Let $l(n)$ denote the number of packets in the $(\sum_{\tau=-t-b}^{-1} 
X_F(\tau))$-th frame that arrives in the interval $[-t-b,-b-1]$. We want to show 
that for any sample path in the set $\mE_t \cap \mE^\gamma_t$, we have
\begin{equation}
\label{eq:ln}
l(n) > f(n)-H\sqrt{n}. 
\end{equation}
Now, recall that 
for every sample path in set $\mE_t$, there are two cases: either 
$\mE^{\alpha_1}_t$ or $\mE^{\alpha_2}_t$ occurs. We consider these 
two cases separately.

{\bf Case 1)} Suppose $\mE^{\alpha_1}_t$ occurs, i.e., $A_F(-t-b,-b-1) > X_F(-t-b,-1)$. 
Then, from (\ref{eq:afe}) and (\ref{eq:xf}), we have $\left \lceil \frac {A(-t-b,-b-1)} 
{n_0} \right \rceil > \sum_{\tau=-t-b}^{-1} X_F(\tau)$. Eq. (\ref{eq:afe}) also implies
that all the frames that arrive in the interval $[-t-b,-b-1]$, except the last frame, 
are fully filled (with $n_0$ packets). Since the $(\sum_{\tau=-t-b}^{-1} X_F(\tau))$-th 
frame is not the last frame, we have $l(n) = n_0 = n - H\sqrt{n} > f(n) - H\sqrt{n}$.

{\bf Case 2)} Suppose $\mE^{\alpha_2}_t$ occurs, i.e., $A_F(-t-b,-b-1) = X_F(-t-b,-1)$
and $U(0)=1$. Then, from (\ref{eq:afe}) and (\ref{eq:xf}), we have $\left \lceil \frac 
{A(-t-b,-b-1)} {n_0} \right \rceil = \sum_{\tau=-t-b}^{-1} X_F(\tau)$. This implies
that the $(\sum_{\tau=-t-b}^{-1} X_F(\tau))$-th frame is the last frame. Now, suppose
that $l(n) \le f(n)-H\sqrt{n}$. Then, the last success in the interval of $[-t-b,-1]$ 
would have completely cleared the $l(n)$ packets as they are among the oldest $f(n)$ 
packets. This implies $U(0)=0$, which contradicts with the assumption that $\mE^{\alpha_2}_t$ 
occurs. Therefore, we must have $l(n) > f(n)-H\sqrt{n}$.

Combining both cases, we show that (\ref{eq:ln}) holds for any sample path in set 
$\mE_t \cap \mE^\gamma_t$.

Recall that for any sample path in set of $\mE_t \cap \mE^\gamma_t$, all the frames 
that arrive in the interval of $[-t-b,-b-1]$, except the last frame, are fully filled 
(with $n_0$ packets), and that $\left \lceil \frac {A(-t-b,-b-1)} {n_0} \right \rceil 
\ge \sum_{\tau=-t-b}^{-1} X_F(\tau)$. Hence, we have
\[
\begin{split}
A(&-t-b,-b-1) \\
&\ge n_0 (\sum_{\tau=-t-b}^{-1} X_F(\tau) - 1) + l(n) \\
&= n(\sum_{\tau=-t-b}^{-1} X_F(\tau) - 1) + l(n) - H\sqrt{n}(\sum_{\tau=-t-b}^{-1} X_F(\tau) - 1) \\
&> n(\sum_{\tau=-t-b}^{-1} X_F(\tau) - 1),
\end{split}
\]
for all $n \ge N_0$, where the last inequality is from (\ref{eq:ln}), 
our choice of $N_0$ such that $f(n)>H\sqrt{n} (t^*+b)$ for all $n \ge 
N_0$, and $\sum_{\tau=-t-b}^{-1} X_F(\tau) \le t^* + b$. Hence, we have 
$\frac {A(-t-b,-b-1)} {n} + 1 - \sum_{\tau=-t-b}^{-1} X_F(\tau) > 0$. 
This implies $\mE_t \cap \mE^\gamma_t \subseteq \mE^\beta_t$ for all 
$n \ge N_0$, and thus (\ref{eq:sp}) holds. Therefore,
\[
\Prob (\mE_t) \le 1 - \Prob (\mE^\gamma_t) + \Prob (\mE^\beta_t).
\]
Along with the choice of $h$ (as chosen earlier), we can show that 
there exist $N_3>0$ and $C_1>0$ such that for all $n \ge N_3$, 
\begin{equation}
\label{eq:palpha}
\Prob (\mE^\gamma_t) > 1 - C_1 t e^{-n I^\epsilon_0} \ge 1 - C_1 t^* e^{-n I^\epsilon_0}.
\end{equation}
Here, we do not duplicate the detailed proof for (\ref{eq:palpha}), 
and refer the interested readers to \cite{sharma11b} for details.

Next, we calculate the probability that event $\mE^\beta_t$ occurs. 
Specifically, we want to show that there exists a finite $N_1>0$ such 
that for all $n \ge N_1$ and for all $t \le t^*$, we have
\[
\Prob (\mE^\beta_t) \le e^{g(n)} e^{-n I^\epsilon_0}.
\]

Recall that $I_{AG}(t,x)=\liminf_{n \rightarrow \infty} \frac {-1}{n} 
\log \Prob(A(-t+1,0) > n(t+x))$ and $I_{AG}(x) = \inf_{t>0} I_{AG}(t,x)$. 
Hence, for any fixed $\epsilon > 0$, there exists a finite $N_4$ such 
that for all $n \ge N_4$, we have 
\begin{equation}
\label{eq:uba}
\begin{split}
\Prob(A(-t+1,0) > n(t+x)) & \le e^{-n(I_{AG}(t,x)-\epsilon)} \\
& \le e^{-n(I_{AG}(x)-\epsilon)}.
\end{split}
\end{equation}

We next calculate an upper bound on the probability that during interval 
$[-t-b,-1]$, there are exactly $t+a$ successes, for some $a \le b$. Recall 
from Lemma~\ref{lem:success} that for all $n \ge N_X$, the probability of 
a success in each time-slot is no smaller than $1-2(1-q)^{n-f(n)\log^2n}$. 
Hence, we have 
\begin{equation}
\label{eq:ubs}
\begin{split}
\Prob ( & \sum_{\tau=-t-b}^{-1} X_F(\tau) = t+a) \\
& \le \binom{t+b} {t+a} (2(1-q)^{n-f(n)\log^2n})^{b-a} \\
& \le 2^{t+b} 2^{b-a} \left(\frac {1} {1-q} \right)^{(b-a)f(n)\log^2n} e^{-n(b-a) \log \frac {1} {1-q}} \\
& \le 2^{t+2b} \left(\frac {1} {1-q} \right)^{bf(n)\log^2n} e^{-n(b-a) I_X}.
\end{split}
\end{equation}
It is easy to observe that the right hand side is a monotonically increasing 
function in $a$.

We choose $N_5>0$ such that $(t^*+b+1) 2^{t^*+2b} \left(\frac {1} {1-q} 
\right)^{b f(n)\log^2n} \le e^{g(n)}$ for all $n \ge N_5$, and choose $N_1 
\triangleq \max \{N_0, N_X, N_3, N_4, N_5 \}$. Using the results from (\ref{eq:uba})
and (\ref{eq:ubs}), we have that for all $n \ge N_1$,
\[
\begin{split}
\Prob (& \mE^\beta_t) \\
& = \Prob ( \frac {A(-t-b,-b-1)} {n} + 1 - \sum_{\tau=-t-b}^{-1} X_F(\tau) > 0 ) \\
& = \sum_{a=0}^{t+b} \Prob ( \sum_{\tau=-t-b}^{-1} X_F(\tau) = a ) \Prob ( A(-t-b,-b-1) > (a-1)n)  \\
& \le (t+b+1) \max_{0 \le a \le t+b} \{ \Prob ( \sum_{\tau=-t-b}^{-1} X_F(\tau) = a ) \\
&~~~~~ \times \Prob ( A(-t-b,-b-1) > (a-1)n) \} \\
& \le (t+b+1) \max \{ \max_{a \in \{0,1, \dots, t \}} \{\Prob ( \sum_{\tau=-t-b}^{-1} X_F(\tau) = a )\}, \\
&~~~~~ \max_{a \in \{1, \dots, b \}} \{ \Prob ( \sum_{\tau=-t-b}^{-1} X_F(\tau) = t+a ) \\
&~~~~~ \times \Prob ( A(-t-b,-b-1) > (t+a-1)n) \} \} \\
& \stackrel{(a)}\le (t+b+1) \max \{ 2^{t+2b} \left(\frac {1} {1-q} \right)^{b f(n) \log^2 n} e^{-nb I_X}, \\
&~~~~~ \max_{a \in \{1, \dots, b \}} \{ \Prob ( \sum_{\tau=-t-b}^{-1} X_F(\tau) = t+a ) \\
&~~~~~ \times \Prob ( A(-t-b,-b-1) > (t+a-1)n) \} \} \\
& \stackrel{(b)}\le (t+b+1) \max \{ 2^{t+2b} \left(\frac {1} {1-q} \right)^{b f(n) \log^2 n} e^{-nb I_X}, \\
&~~~~~ \max_{a \in \{1, \dots, b \}} \{ 2^{t+2b} \left(\frac {1} {1-q} \right)^{b f(n) \log^2 n} e^{-n(b-a) 
I_X} e^{-n (I_{AG}(a-1)-\epsilon)}  \} \} \\
& \le (t+b+1) 2^{t+2b} \left(\frac {1} {1-q} \right)^{b f(n) \log^2 n} \\
&~~~~~ \times \max \{e^{-nb I_X}, e^{-n \min_{a \in \{1,\dots,b \}} \{I_{AG}(a-1) - \epsilon + (b-a) I_X\}} \} \\
& \stackrel{(c)}\le (t^*+b+1) 2^{t^*+2b} \left(\frac {1} {1-q} \right)^{b f(n) \log^2 n} \\
&~~~~~ e^{-n \min \{ b I_X, \min_{c \in \{0,1,\dots,b-1 \}} \{I_{AG}(b-1-c) - \epsilon + c I_X\} \} } \\
& \stackrel{(d)}\le e^{g(n)} e^{-n I^\epsilon_0},
\end{split}
\]
where (a) is from the monotonicity of the right hand side of (\ref{eq:ubs}),
(b) is from (\ref{eq:uba}) and (\ref{eq:ubs}), (c) is from changing variable
by setting $c=b-a$, and (d) is from our choice of $N_1$.

Summing over $t=1$ to $t^*$, we have 
\[
P_1 = \sum_{t=1}^{t^*} \Prob (\mE_t) \le \sum_{t=1}^{t^*} (1 - \Prob (\mE^\gamma_t) + \Prob (\mE^\beta_t)) 
\le (C_1 + e^{g(n)}) t^* e^{-n I^\epsilon_0},
\]
for all $n \ge N_1$.

{\bf Part 2:}
We want to show that there exists a finite $N_2>0$ such that for all $n \ge N_2$, we have
\[ 
P_2 \le 4 e^{-n I^\epsilon_0}.
\]

As in the proof for Theorem~2 of \cite{sharma11b}, for any fixed real number 
$\hat{p} \in (p,1)$, we consider the arrival process $\hat{A}(\cdot)$, by 
adding extra dummy arrivals to the original arrival process $A(\cdot)$. The 
resulting arrival process $\hat{A}(\cdot)$ is simple, and has the following 
property:
\[
\hat{A}(\tau) = \left\{
\begin{array}{ll}
\hat{p} n, & \text{if}~ A(\tau) \le \hat{p} n,\\
L n, & \text{if}~ A(\tau) > \hat{p} n.
\end{array}
\right. 
\]
Hence, if we can find an upper bound on $\hat{A}_F(-t-b,-b-1)$, by our 
construction, then it is also an upper bound on $A_F(-t-b,-b-1)$.

Consider any $t \ge t^*$. Let $B=\{b_1, b_2, \dots, b_{|B|}\}$ be the set 
of time-slots in the interval from $-t-b$ to $-b-1$ when $\hat{A}(\tau) = Ln$. 
Given $L(-b)=-t-b-1$, from Lemma~\ref{lem:full}, we have that,
\[
\begin{split}
\hat{A}_F( & -t-b,-b-1) \\
& \le \sum_{r=1}^{|B|-1} \left \lceil \frac {\hat{A}(b_r + 1, b_{r+1}-1)} {n_0} 
\right \rceil + \sum_{r=1}^{|B|} \left \lceil \frac {\hat{A}(b_r,b_r)} {n_0} \right \rceil \\
&~~~~~ + \left \lceil \frac {\hat{A}(-t-b,b_1 - 1)} {n_0} \right \rceil
+ \left \lceil \frac {\hat{A}(b_{|B|} + 1, -b-1)} {n_0} \right \rceil \\
& \le \sum_{r=1}^{|B|-1} \frac {\hat{A}(b_{r+1} - 1 - b_r)\hat{p}n} {n_0} + |B|-1
+ \sum_{r=1}^{|B|} \frac {Ln} {n_0} + |B| \\
&~~~~~ + \frac {(b_1+t+b) \hat{p} n} {n_0} + 1 + \frac {(-b-1 - b_{|B|})\hat{p} n} {n_0} + 1 \\
& \le \frac {(t-|B|)\hat{p} n} {n_0} + |B| \frac {Ln} {n_0} + 2 |B| + 1 \\
& \le \frac {n} {n_0} (\hat{p}t + (L+2)|B| + 1).
\end{split}
\]

From Assumption~\ref{ass:arr_ld} on the arrival process we know that for 
large enough $n$ and $t$, $|B|$ can be made less than an arbitrarily small
fraction of $t$. Further, we can show that for $n \ge (\frac {H(2+\hat{p})} 
{1-\hat{p}})^2, t > \frac {18(1+\hat{p})} {(2+\hat{p})(1-\hat{p})}$ and 
$|B|<\frac {1-\hat{p}}{6(L+2)}t$, we have $A_F(-t-b,-b-1) < (\frac {2+\hat{p}}{3})t-1$. 
This is derived by substituting the values of $n,t$ and $|B|$ in the equation above,
\[
\begin{split}
A_F( & -t-b,-b-1) \\
& \le \hat{A}_F(-t-b,-b-1) \\
& \le \frac {n} {n_0} (\hat{p}t + (L+2)|B| + 1) \\
& < \frac {2+\hat{p}} {1+2\hat{p}} \left( \hat{p}t+\frac {1-\hat{p}} {6}t + 
(\frac {1-\hat{p}} {6}t - \frac {1+2\hat{p}}{2+\hat{p}}) \right) \\
& = \frac {2+\hat{p}} {1+2\hat{p}} \left( \frac {1+2\hat{p}} {3}t - 
\frac {1+2\hat{p}}{2+\hat{p}} \right ) \\
& \le (\frac {2+\hat{p}} {3}) t - 1.
\end{split}
\]
Then, it follows that
\begin{equation}
\label{eq:prob_arr}
\begin{split}
\Prob ( & A_F(-t-b,-b-1) \ge (\frac {2+\hat{p}} {3}) t - 1, L(-b)=-t-b-1) \\
& = 1 - \Prob(A_F(-t-b,-b-1) < (\frac {2+\hat{p}} {3}) t - 1, L(-b)=-t-b-1) \\
& \le 1 - \Prob (|B| \le \frac {1-\hat{p}} {6(L+2)} t) \\
& = \Prob (|B| > \frac {1-\hat{p}} {6(L+2)} t) \\
& \le e^{-ntI_B(\hat{p}-p,\frac {1-\hat{p}} {6(L+2)})},
\end{split}
\end{equation}
for all $n \ge N_6 \triangleq \max \{ N_B(\hat{p}-p, \frac {1-\hat{p}} {6(L+2)}),
(\frac {2(2+\hat{p})} {1-\hat{p}})^2 \}$ and $t \ge T_1$, where the last inequality
is from Assumption~\ref{ass:arr_ld} and (\ref{eq:t1}).

We now state a lemma that will be used in the proof.
\begin{lemma}
\label{lem:brv}
Let $X_i$ be a sequence of binary random variables satisfying
\[
\Prob (X_i=0) < c(n) e^{-nd}, ~\text{for all}~ i,
\]
where $c(n)$ is a polynomial in $n$ of finite degree. Let $N^{\prime}$
be such that $c(n) < e^{\frac {nd} {2}}$ for all $n \ge N^{\prime}$.
Then, for any real number $a \in (0,1)$, we have
\[
\Prob \left (\sum_{i=1}^{t} X_i < (1-a)t \right) \le e^{- \frac {tnad} {3}} 
\]
for all $n \ge \max \{ \frac {12} {ad}, N^{\prime} \}$.
\end{lemma}

\begin{proof}
The proof follows immediately from Lemma~1 of \cite{sharma11b}.
\end{proof}

Moreover, we know from Lemma~\ref{lem:success} that for each $\tau$, 
$X_F(\tau)=0$ with probability less than $2(1-q)^{n-f(n) \log^2 n} 
= 2(\frac {1} {1-q})^{f(n) \log^2 n} e^{-nI_X}$ for all $n \ge N_X$. 
Choose $N_7$ such that $2(\frac {1} {1-q})^{f(n) \log^2 n} < e^{\frac 
{nI_X}{2}}$ for all $n \ge N_7$. Hence, choosing $N_8 \triangleq \max 
\{N_X, N_7, \frac {36} {(1-\hat{p})I_X} \}$ and using Lemma~\ref{lem:brv}, 
we have 
\begin{equation}
\label{eq:prob_suc}
\begin{split}
\Prob ( & X_F(-t-b,-1)<(\frac {2+\hat{p}} {3})t, L(-b)=-t-b-1) \\
& \le \Prob (X_F(-t-b,-1)<(\frac {2+\hat{p}} {3}) (t+b), L(-b)=-t-b-1) \\
& \le e^{-n(t+b)(\frac {1-\hat{p}}{9})I_X} \\
& \le e^{-nt\frac {(1-\hat{p})I_X}{9}}, 
\end{split}
\end{equation}
for all $n \ge N_8$ and $t>0$.

From (\ref{eq:prob_arr}) and (\ref{eq:prob_suc}), we have that for all $n \ge 
N_9 \triangleq \max \{ N_6, N_8 \}$ and $t \ge T_1$,
\[
\begin{split}
\Prob ( & A_F(-t-b,-b-1)+1-X_F(-t-b,-1)>0, L(-b)=-t-b-1) \\
& \le 1-(1-e^{-nt(\frac {1-\hat{p}}{9})I_X})(1-e^{-ntI_B(\hat{p}-p,\frac {1-\hat{p}} {6(L+2)})}) \\
& \le 2 e^{-ntI_{BX}},
\end{split}
\]
where the last inequality is from (\ref{eq:ibx}).

Then, summing over all $t \ge t^*$, we have that for all $n \ge N_2 \triangleq 
\max \{ N_9, \left \lceil \frac {\log 2} {I_{BX}} \right \rceil \}$,
\[
\begin{split}
P_2 & = \sum_{t=t^*}^{\infty} \Prob(L(-b)=-t-b-1, \mE^{\alpha}_t) \\
& \le \sum_{t=t^*}^{\infty} \Prob(L(-b)=-t-b-1, \\
&~~~~~~~~~~~ A_F(-t-b,-b-1) + 1 > X_F(-t-b,-1)) \\
& \le \sum_{t=t^*}^{\infty} 2 e^{-ntI_{BX}} \\
& \le \frac {2 e^{-nt^*I_{BX}}} {1-e^{-nI_{BX}}} \\
& \stackrel{(a)}\le 4 e^{-n t^* I_{BX}} \\
& \stackrel{(b)}\le 4 e^{-n I^\epsilon_0},
\end{split}
\]
where (a) is from our choice of $N_2$, and (b) is from (\ref{eq:tstar}).

Combining both parts, the result of the theorem then follows.

\section{Proof of Lemma~\ref{lem:dssg-dom}}  \label{app:lem:dssg-dom}
Consider two queueing systems, $\bar{Q}_1$ and $\bar{Q}_2$, both of 
which have the same arrival and channel realizations. We assume that 
$\bar{Q}_1$ adopts the D-SSG policy and $\bar{Q}_2$ adopts the G-FBS 
policy (with any value of $h$). We define the weight of a packet $p$ 
in time-slot $t$ as its delay, i.e., $w(p) = t - t_p$, where $t_p$ 
denotes the time when packet $p$ arrives to the system. Note that 
different packets (in the same queue or in different queues) may have 
the same delay. In order to make each packet in the system have a unique 
weight, we redefine the weight of a packet $p$ as $\hat{w}(p) = t - t_p 
+ \frac {n+1-q_p} {n+1} + \frac {L+1-x_p} {(L+1) (n+1)}$, where $q_p$ 
denotes the index of the queue that contains packet $p$ and $x_p$ 
denotes that packet $p$ is the $x_p$-th arrival to queue $q_p$ in 
time-slot $t_p$. For two packets $p_1$ and $p_2$, we say $p_1$ is 
older than $p_2$ if $\hat{w}(p_1) > \hat{w} (p_2)$. Note that as in 
\cite{sharma11,sharma11b}, we use weight $\hat{w}(\cdot)$ instead of 
$w(\cdot)$ for ease of analysis only.

Let $R_i(t)$ denote the set of packets present in system $\bar{Q}_i$ 
at the end of time-slot $t$, for $i=1,2$. Then, it suffices to show 
that $R_1(t) \subseteq R_2(t)$ for all time-slot $t$. We let $A(t)$ 
denote the set of packets that arrive in time-slot $t$. Let $X_i(t)$ 
denote the set of packets that depart system $\bar{Q}_i$ at time $t$, 
for $i=1,2$. Hence, we have
\[
R_i(t+1) = (R_i(t) \cup A(t+1)) \backslash X_i(t+1),~\text{for}~ i=1,2.
\]

We then proceed our proof by contradiction. Suppose that $R_1(t) 
\nsubseteq R_2(t)$ for some time-slot $t$. Without loss of generality,
we assume that time-slot $\tau$ is the first time such that $R_1(\tau) 
\nsubseteq R_2(\tau)$ occurs. Hence, there must exist a packet, say $p$, 
such that $p \in R_1(\tau)$ but $p \notin R_2(\tau)$. Because $\tau$ is 
the first time when such an event occurs, packet $p$ must depart from 
system $\bar{Q}_2$ in time-slot $\tau$ while it is not served in system
$\bar{Q}_1$, i.e., $p \in X_2(\tau)$ and $p \notin X_1(\tau)$.
Let $B_i(v)$ denote the set of packets in $R_i(\tau-1) \cup A(\tau)$
with weight greater than $v$, for $i=1,2$. Clearly, we have $B_1(v) 
\subseteq B_2(v)$ for all $v$, as $R_1(\tau-1) \subseteq R_2(\tau-1)$ 
by assumption. Let $\mS_i(v)$ denote the set of servers that are chosen 
to serve packets in $B_i(v)$, for $i=1,2$. 

Now, consider time-slot $\tau$. Given that packet $p$ is not served in system 
$\bar{Q}_1$, we want to show that packet $p$ cannot be served in system 
$\bar{Q}_2$ either. From the hypothesis, we know that packet $p$ must be 
disconnected from any server in set $\mS \backslash \mS_1(\hat{w}(p))$, 
otherwise due to the operations of D-SSG, packet $p$ must be served by some 
server in set $\mS \backslash \mS_1(\hat{w}(p))$ in system $\bar{Q}_1$. 
Note that both systems have the same channel realizations. Hence, packet
$p$ cannot be served by any server in set $\mS \backslash \mS_1(\hat{w}(p))$ 
in system $\bar{Q}_2$ either. Next, we show that packet $p$ is not served 
by any server in set $\mS_1(\hat{w}(p))$ in system $\bar{Q}_2$.
Consider any server $S_{j(r)} \in \mS_1(\hat{w}(p))$, which 
serves some packet $x_r \in B_1(\hat{w}(p))$ in system $\bar{Q}_1$ (under 
D-SSG). Then, one of the following must occur in system $\bar{Q}_2$ (under 
G-FBS, which runs D-SSG over the oldest $n$ packets): either server $S_{j(r)}$ 
also serves packet $x_r$, or server $S_{j(r)}$ is allocated to serve some 
packet with a larger weight than packet $x_r$ due to the operations of D-SSG 
and the fact that $B_1(\hat{w}(p)) \subseteq B_2(\hat{w}(p))$. This implies 
that packet $p$ is not served by any server in set $\mS_1(\hat{w}(p))$ in 
system $\bar{Q}_2$. Therefore, packet $p$ is not served by any server in 
system $\bar{Q}_2$. This leads to a contradiction that packet $p$ departs 
from system $\bar{Q}_2$ in time-slot $\tau$, and thus completes the proof.

\section{Proof of Theorem~\ref{thm:dssg-to}}  \label{app:thm:dssg-to}
We first restate the definition of a class of throughput-optimal MWF
policies, and then show that D-QSG is an MWF policy.

Recall that $Q_i(t)$ denotes the length of queue $Q_i$ at the beginning
of time-slot $t$ immediately after packet arrivals, $Z_{i,l}(t)$ denotes
the delay of the $l$-th packet of $Q_i$ at the beginning of time-slot $t$, 
$W_i(t) = Z_{i,1}(t)$ denotes the delay of the HOL packet of $Q_i$ at the 
beginning of time-slot $t$, and $C_{i,j}(t)$ denotes the connectivity 
between $Q_i$ and $S_j$ in time-slot $t$. Also, let $\mS_j(t)$ 
denote the set of queues being connected to server $S_j$ in time-slot 
$t$, i.e., $\mS_j(t) = \{1 \le i \le n ~|~ C_{i,j}(t)=1 \}$, and let
$\Gamma_j(t)$ denote the subset of queues in $\mS_j(t)$ that have the 
largest weight in time-slot $t$, i.e., $\Gamma_j(t) \triangleq \{i \in 
\mS_j(t) ~|~ W_i(t) = \max_{l \in \mS_j(t)} W_l(t) \}$. 
Let $i(j,t)$ be the index of the queue that is served by server $S_j$ in 
time-slot $t$, under a scheduling policy $\mathbf{P}$. Policy $\mathbf{P}$ 
is said to be an MWF policy if there exists a constant $M>0$ such that, in 
any time-slot $t$ and for all $j \in \{1,2,\dots,n \}$, queue $Q_{i(j,t)}$ 
satisfies that $W_{i(j,t)}(t) \ge Z_{r,M}(t)$ for all $r \in \Gamma_j(t)$ 
such that $Q_r(t) \ge M$. It has been proven that any MWF policy is 
throughput-optimal \cite{ji13a}.

Then, it remains to show that D-SSG is an MWF policy. Let $M=n$. We want 
to show that in any time-slot $t$ and for all $j \in \{1,2,\dots,n \}$, 
D-SSG allocates server $S_j$ to serve queue $Q_{i(j,t)}$, which satisfies 
that $W_{i(j,t)}(t) \ge Z_{r,n}(t)$ for all $r \in \Gamma_j(t)$ such that 
$Q_r(t) \ge n$.

Consider any round $j$ in any time-slot $t$, where server $S_j$ needs to be 
allocated. Also, consider any $r \in \Gamma_j(t)$ such that $Q_r(t) \ge n$. 
Suppose that server $S_j$ is allocated to serve queue $Q_{i(j,t)}$, then we have 
that $W_{i(j,t)}(t) \ge W^{j-1}_{i(j,t)}(t) \ge W^{j-1}_r(t) \ge Z_{r,n}(t)$, 
where the first inequality is because the HOL packet of queue $Q_{i(j,t)}$ at 
the beginning of time-slot $t$ must have a delay no smaller than the HOL packet 
of queue $Q_{i(j,t)}$ at the beginning of the $j$-th round, the second inequality 
is due to the operations of D-SSG, and the last inequality is because the HOL packet 
of queue $Q_r$ at the beginning of the $j$-th round must have a delay no smaller 
than the $n$-th packet in queue $Q_r$ at the beginning of time-slot $t$. This implies 
that the sufficient condition is satisfied under D-SSG. Therefore, D-SSG is an MWF 
policy and is thus throughput-optimal.

\bibliographystyle{IEEEtran}
\bibliography{ofdm}

\end{document}